\documentclass{article}
\usepackage{amsmath,amsfonts,amsthm}
\usepackage{geometry}
\newtheorem{thm}{Theorem}[section]
\newtheorem{lemma}[thm]{Lemma}
\newtheorem{prop}[thm]{Proposition}
\newtheorem{coro}[thm]{Corollary}

\theoremstyle{remark}
\newtheorem{rem}[thm]{Remark}
\newtheorem{ex}[thm]{Example}

\newcommand{\vertiii}[1]{{\left\vert\kern-0.25ex\left\vert\kern-0.25ex\left\vert #1 
    \right\vert\kern-0.25ex\right\vert\kern-0.25ex\right\vert}}
\def\Me{\mathcal M}
\def\Se{\mathcal S}

\def\Ne{\mathcal N}
\def\Ha{\mathcal H}

\def\Fe{\mathcal F}

\def\Tr{\mathrm{Tr}\,}
\def\states{\mathfrak S}
\def\<{\langle}
\def\>{\rangle}
\def\supp{\mathrm{supp}}
\begin{document}

\title{R\'enyi relative entropies and noncommutative $L_p$-spaces}
\author{Anna Jen\v cov\'a \thanks{jenca@mat.savba.sk}\\ \small \emph{Mathematical Institute, Slovak Academy of Sciences}\\
\small { \v Stef\'anikova 49, 814 73 Bratislava, Slovakia}}
\date{}

\maketitle

\begin{abstract} We propose an extension of the sandwiched R\'enyi relative $\alpha$-entropy to normal positive
functionals on arbitrary von Neumann algebras, for the values $\alpha>1$. For this, we use Kosaki's definition of
noncommutative $L_p$-spaces with respect to a state. We show that these extensions coincide with the previously defined
Araki-Masuda divergences [M. Berta et al., {\em Annales Henri Poincar{\'e}}, 19:1843--1867, 2018] and prove some of their properties,  in particular the  data processing inequality with respect to positive normal unital maps. 
As a consequence, we obtain monotonicity of the  Araki relative entropy with respect to such maps,
 extending the results of [A.~M\"uller-Hermes and D.~Reeb. {\em Annales Henri Poincar\'e}, {18}:{1777--1788}, {2017}] to  arbitrary von Neumann algebras.  It is also shown that equality in data processing inequality characterizes sufficiency (reversibility) of quantum channels.
\end{abstract}

\section{Introduction}

The classical R\'enyi relative entropies were introduced by an axiomatic approach in \cite{renyi1961onmeasures},  as the unique family of divergences satisfying certain natural properties. As it turned out, these quantities play a central role in many information-theoretic tasks, see e.g. \cite{csiszar1995generalized} for an overview. A  straightforward quantum generalization is given by standard quantum R\'enyi relative $\alpha$-entropies, defined for density matrices $\rho,\sigma$ as
\[
 D_\alpha(\rho\|\sigma)= \left\{\begin{array}{cc}    \frac1{\alpha-1} \log\left( \Tr \rho^\alpha\sigma^{1-\alpha}\right)& \mathrm{if }\  \alpha\in (0,1)\mbox{  or } \supp(\rho)\subseteq \supp(\sigma)\\ & \\
 \infty & \mathrm{otherwise,}
 \end{array}
 \right. 
\]
where $\supp(\rho)$ denotes the support of $\rho$ and $\alpha>0$, $\alpha\ne 1$.
These quantities share the useful properties of the classical R\'enyi relative entropy, but not for all values of the parameter $\alpha$. 
In particular, for a quantum channel $\Phi$, the data processing inequality (DPI)
\begin{equation}\label{eq:dpi}
 D_\alpha(\Phi(\rho)\|\Phi(\sigma))\le  D_\alpha(\rho\|\sigma)
\end{equation}
 holds for $\alpha$ in the range $(0,2]$, \cite{petz1984quasi,hmpb2011divergences}. 
Moreover, for  $\alpha\in (0,1)$ the standard R\'enyi relative entropies appear as error exponents and cutoff rates in hypothesis testing \cite{ansv2008chernoff,hmo2008exponents, himo2011onthequantum}.

Another quantum version of  R\'enyi relative entropy was introduced in \cite{wwy2014strong, mldsft20130nquantum}. It is the 
sandwiched R\'enyi relative $\alpha$-entropy, defined as 
\begin{align}\label{eq:sandwiched}
\tilde D_\alpha(\rho \| \sigma)=\left\{\begin{array}{cc} \frac{1}{\alpha-1}\log \Tr \left[\left( \sigma^{\frac{1-\alpha}{2\alpha}}\rho
\sigma^{\frac{1-\alpha}{2\alpha}}\right)^\alpha\right] & \mathrm{if }\ \supp(\rho)\subseteq \supp(\sigma)
\\  & \\
\infty & \mathrm{otherwise}
 \end{array}
 \right.
\end{align}
for $\alpha>0$, $\alpha\ne 1$. The sandwiched entropies satisfy DPI for $\alpha\in [1/2,1)\cup (1,\infty)$, \cite{mldsft20130nquantum, wwy2014strong, beigi2013sandwiched, frli2013monotonicity}. 
For $\alpha>1$, $\tilde D_\alpha$ have an operational meaning  as strong converse exponents in quantum hypothesis testing and channel coding, \cite{moog2015quantum,moog2014strong}. Moreover, 
both $ D_\alpha$ and $\tilde D_\alpha$ yield the Umegaki relative entropy 
\[
D_1(\rho\|\sigma)=\left\{\begin{array}{cc}  \Tr \rho(\log(\rho)-\log(\sigma)) & \mathrm{if }\ \supp(\rho)\subseteq \supp(\sigma)\\ & \\
 \infty & \mathrm{otherwise}
 \end{array}
 \right.
\]
in the limit as $\alpha\to 1$. On the other hand, in the limit $\alpha\to \infty$, $\tilde D_\alpha$ gives the relative max entropy 
\begin{equation}\label{eq:maxent}
\tilde D_\infty(\rho\|\sigma)= \log(\inf\{\lambda>0,\ \rho\le \lambda \sigma\},
\end{equation}
see \cite{wwy2014strong, mldsft20130nquantum} for the proofs of these properties.
\begin{rem} Both $D_\alpha$ and $\tilde D_\alpha$ are contained in the family of entropic pressure functionals introduced in \cite[Section 3.3]{jopp2012entropic} as a tool for studying entropic fluctuations in quantum statistical mechanics. The same family of functionals in the context of quantum information theory was studied in \cite{auda2015alpha}.

\end{rem}

Let $(\Phi,\rho,\sigma)$ be a triple consisting of a quantum  channel  $\Phi$ and a  pair of states $\rho$, $\sigma$ on the input space of $\Phi$. A channel $\Psi$ satisfying $\Psi\circ\Phi(\rho)=\rho$ and  $\Psi\circ\Phi(\sigma)=\sigma$ is called a recovery map for $(\Phi,\rho,\sigma)$. If a recovery map exists, we say that the  channel 
$\Phi$ is sufficient (or reversible) with respect to $\{\rho,\sigma\}$. This terminology was introduced in \cite{petz1986sufficient, petz1988sufficiency}, by analogy with the classical notion of a sufficient statistic. Clearly, if $\Phi$ is sufficient with respect to 
$\{\rho,\sigma\}$, equality must be attained in DPI. It is much less obvious that the opposite implication holds in some cases. This was first observed in \cite{petz1986sufficient,petz1988sufficiency} for $D_1$  and $ D_{1/2}$ and  later extended to a large class of quantum divergences, including  $D_\alpha$ with $\alpha\in (0,2)$, \cite{hmpb2011divergences}. 
The same property for  $\tilde D_\alpha$ with $\alpha>1$ was proved in \cite{jencova2016preservation}.

Quantum versions of relative entropies are usually studied in the finite dimensional setting. Nevertheless, the standard version $ D_\alpha$ is derived from 
the quasi-entropies \cite{petz1984quasi}, which were defined  in \cite{petz1985quasi} also in the more general context of von Neumann algebras.
Moreover, the entropies $D_\alpha$ have similar properties and play a similar role in quantum hypothesis testing in this setting, \cite{jops2012quantum}. 
A definition of sandwiched R\'enyi entropies for states on von Neumann algebras was recently 
 proposed in \cite{berta2018renyi}. These entropies are called the Araki-Masuda divergences and are  based on the Araki-Masuda definition of non-commutative $L_p$ spaces with respect to a state. It is conjectured that these quantities characterize the strong converse exponents in binary quantum hypothesis testing, as in the finite dimensional case.

 The aim of the present work is to propose a von Neumann algebraic extension of $\tilde D_\alpha$ for $\alpha>1$ using the interpolating family of Kosaki's  $L_p$-spaces, \cite{terp1982interpolation,kosaki1984applications}. This approach was inspired by the work by Beigi \cite{beigi2013sandwiched}, where a similar family of norms (in finite dimensions)  was used to prove DPI for $\tilde D_\alpha$ with 
$\alpha>1$. It was later observed \cite{mhre2015monotonicity} that this method works even for positive trace-preserving maps and taking the limit $\alpha\to 1$ implies that the quantum relative entropy is monotone under such mappings. This important result was also extended to density operators on infinite dimensional Hilbert spaces.  The  framework of interpolation norms was also used  in \cite{jencova2016preservation} to show that in finite dimensions, equality in DPI for $\tilde D_\alpha$, $\alpha>1$ implies sufficiency of the channel.

As one of the main results, we prove that the proposed quantities coincide 
 with the Araki-Masuda divergences of \cite{berta2018renyi}. This  was independently proved by Hiai \cite{hiai2017unpublished}, using  different methods.
For normal states  $\psi$, $\varphi$ of an arbitrary von Neumann algebra $\Me$, we  further  prove the following properties of $\tilde D_\alpha$: 
 \begin{enumerate}
\item[(a)] Positivity: $\tilde D_\alpha(\psi\|\varphi)\ge 0$, with equality if and only if $\psi=\varphi$.
\item[(b)] Monotonicity: if $\psi\ne\varphi$ and $\tilde D_\alpha(\psi\|\varphi)<\infty$ for some $\alpha>1$, then 
the function $\alpha'\mapsto \tilde D_{\alpha'}(\psi\|\varphi)$ is strictly increasing for $\alpha'\in (1,\alpha]$.
\item[(c)] Limit values: for $\alpha\to 1$, the Araki relative entropy $D_1(\psi\|\varphi)$ is obtained, $\alpha\to \infty$ yields the relative max-entropy $\tilde D_\infty(\psi\|\varphi)$.
\item[(d)] Relation to the standard R\'enyi relative entropy:  for $\alpha>1$,
\[
D_{2-1/\alpha}(\psi\|\varphi)\le \tilde D_\alpha(\psi\|\varphi)\le D_\alpha(\psi\|\varphi).
\]
\item[(e)] Order relations: $\tilde D_\alpha$ can be extended to all positive normal functionals on $\Me$. With this extension, $\psi_0\le \psi$ and $\varphi_0\le\varphi$ imply \[
\tilde D_\alpha(\psi_0\|\varphi)\le \tilde D_\alpha(\psi\|\varphi),\quad 
\tilde D_\alpha(\psi\|\varphi_0)\ge \tilde D_\alpha(\psi\|\varphi).
\]
\item[(f)] Lower semicontinuity: the map $(\psi,\varphi)\mapsto \tilde D_\alpha(\psi\|\varphi)$
 is jointly lower semicontinuous (on the positive part of the predual of $\Me$)
\item[(g)] Generalized mean: let $\psi=\psi_1\oplus \psi_2$, 
$\varphi=\varphi_1\oplus\varphi_2$. Then
\begin{align*}
\exp\{(\alpha-1)\tilde D_\alpha(\psi\|\varphi)\}=\exp&\{(\alpha-1)\tilde D_\alpha(\psi_1\|\varphi_1)\}\\
&+\exp\{(\alpha-1)\tilde D_\alpha(\psi_1\|\varphi_1)\}.
\end{align*}
\item[(h)] Data processing inequality: $\tilde D_\alpha(\Phi(\psi)\|\Phi(\varphi))\le \tilde D_\alpha(\psi\|\varphi)$ holds for any $\alpha>1$ and any  positive trace preserving map $\Phi$. We also give some 
lower and upper bounds on the value of $ \tilde D_\alpha(\psi\|\varphi)- \tilde D_\alpha(\Phi(\psi)\|\Phi(\varphi))$.
 \end{enumerate}
 We also prove a characterization of sufficiency: if $1<\alpha<\infty$ and $\tilde D_\alpha(\psi\|\varphi)$ is finite, then 
 equality in DPI for a 2-positive trace preserving map $\Phi$  implies that $\Phi$  is sufficient with respect to $\{\psi,\varphi\}$.

The properties (a)-(d) and (h) for the Araki-Masuda divergences were shown in \cite{berta2018renyi}. Nevertheless we give 
independent proofs in our setting, which is closely related to the required interpolation techniques. 
Note also that only the second inequality was proved in (d) and complete positivity was required for  (h). Our proof of DPI is close to that of \cite{beigi2013sandwiched} and only positivity is assumed. Note also that (c) and (h) together imply that the relative entropy $D_1(\psi\|\varphi)$ is monotone under positive trace preserving maps between the preduals, which extends the result of \cite{mhre2015monotonicity}.

The outline of the paper is as follows. In Section \ref{sec:nclp}, we introduce the Kosaki's $L_p$-spaces and give an overview of their properties, together with some 
technical results needed later. In Section \ref{sec:sandwiched}, we give the definition of $\tilde D_\alpha$ and prove the equality with Araki-Masuda divergences as well as the properties (a)-(h). The last section deals with sufficiency of maps. 
Some more technical details and  a brief review on the complex interpolation method are given in the Appendices.

\section{Non-commutative $L_p$ spaces with respect to a state} \label{sec:nclp}

Let $\Me$ be a ($\sigma$-finite) von Neumann algebra acting on a Hilbert space $\Ha$ and let $\Me^+$ be the cone of positive elements in $\Me$. We denote the predual by $\Me_*$, its positive part by $\Me_*^+$ and the set of normal states by $\states_*(\Me)$. For $\psi\in \Me_*^+$, we will denote by $s(\psi)$ the support projection  of $\psi$.
For $1\le p\le \infty$, let $L_p(\Me)$ be the 
Haagerup's $L_p$-space over $\Me$ \cite{haagerup1979Lp}, precise definitions and further details on $L_p(\Me)$ can be found in the notes \cite{terp1981lpspaces}.
 
We will use the identification $\Me_*\ni \psi \leftrightarrow h_\psi\in L_1(\Me)$ and the notation $\Tr h_\psi=\psi(1)$ for the trace in $L_1(\Me)$. It this way,
 $\states_*(\Me)$ is identified with the subset of elements in the positive cone $L_1(\Me)^+$ with unit trace. We will also assume the standard form 
$(\lambda(\Me), L_2(\Me), J, L_2(\Me)^+)$ for $\Me$,  see \cite[Thm. 3.6]{terp1981lpspaces}, where $\lambda$ is the left action
\[
\lambda(x): h\mapsto xh, \qquad h\in L_2(\Me),\ x\in \Me
\]
and the conjugation $J$ is defined by $Jh=h^*$, $h\in L_2(\Me)$, see \cite{stzh1979lectures, takesaki2003TOAII} for the definition of a standard form. We  denote
 the inner product in $L_2(\Me)$ by
\begin{equation}\label{eq:Haage_innerp}
(h,k):=\Tr k^*h,\qquad h,k\in L_2(\Me).
\end{equation}
For $k\in L_2(\Me)$, let $\omega_k\in \Me_*^+$ be the linear functional determined by $k$, that is,
\[
\omega_k(a)=(ak,k),\qquad a\in \Me.
\]
For any $\varphi\in \Me_*^+$, $h_\varphi^{1/2}$ is the  unique vector representative of $\varphi$ in the positive cone $L_2(\Me)^+$.

In this section, we describe the noncommutative $L_p$-spaces with respect to a faithful normal state $\varphi$ obtained by complex interpolation. These spaces were defined 
in \cite{kosaki1984applications, trunov1979anoncommutative, zolotarev1982lpspaces} and also in \cite{terp1982interpolation}, where $\varphi$ is allowed to be a weight.  
We will follow the construction by Kosaki, details can be found in  \cite{kosaki1984applications}. 

\subsection{The space $L_\infty(\Me,\varphi)$}
Fix a  faithful normal state $\varphi$ on $\Me$.  To apply the complex interpolation method, we first show that  $\Me$ can be continuously embedded into $L_1(\Me)\simeq  \Me_*$. 
For $x\in \Me$, we put 
\[
h_x:=h_\varphi^{1/2}xh_\varphi^{1/2}.
\]
 By H\"older's inequality \cite[Thm. 23]{terp1981lpspaces}, we have $h_x\in L_1(\Me)$ and $\|h_x\|_1\le \|x\|$. Moreover, $x\mapsto h_x$ is injective and $h_x\in L_1(\Me)^+$ if  and only if $x$ is positive. Note also that for $y\in \Me$, 
\[
\Tr h_xy=\Tr h^{1/2}_\varphi xh^{1/2}_\varphi y= \Tr h^{1/2}_\varphi yh^{1/2}_\varphi x=\Tr h_yx. 
\]
The map $x\mapsto h_x$ is obviously linear and defines a continuous positive embedding of $\Me$ into $L_1(\Me)$.  The image of $\Me$ is the  dense  linear subspace  
\[
L_\infty(\Me, \varphi):=\{h_x,\ x\in \Me\}\subseteq L_1(\Me).
\] 
 The  norm in $L_\infty(\Me, \varphi)$ is introduced as 
\[
\|h_x\|_{\infty,\varphi}:=\|x\|.
\]
The next lemma shows that positive elements in $L_\infty(\Me, \varphi)$  
can be easily characterized. This result is a straightforward consequence of the commutant Radon-Nikodym theorem,  we give a proof for completeness.

\begin{lemma}\label{lemma:linftyplus} 
Let $k\in L_1(\Me)^+$. Then $k=h_x$ for some $x\in \Me^+$ if and only if 
$k\le \lambda h_\varphi$ for some $\lambda>0$. In this case, 
\[
\|k\|_{\infty,\varphi}=\|x\|=\inf\{\lambda>0, k\le\lambda h_\varphi\}.
\]

\end{lemma}

\begin{proof} Let $x\in \Me^+$, then for all $a\in \Me^+$, 
\[
\Tr h_xa\le \|h_xa\|_1=\|xh_\varphi^{1/2}ah_\varphi^{1/2}\|_1\le \|x\|\Tr h_\varphi a
\]
by H\"older's inequality, so that $h_x\le \|x\|h_\varphi$. Conversely, let $0\le k\le \lambda h_\varphi$. By the commutant Radon-Nikodym theorem 
\cite[Section 5.19]{stzh1979lectures}, there is some $x\in \Me$ such that $0\le x \le \lambda$ and for all $y\in \Me$,
\[
\Tr ky=(yh_\varphi^{1/2}, Jxh_\varphi^{1/2})=(yh_\varphi^{1/2}, h_\varphi^{1/2}x^*)=
\Tr xh_\varphi^{1/2}yh_\varphi^{1/2}=\Tr h_\varphi^{1/2}xh_\varphi^{1/2}y.
\]
It follows that $k=h_x$. The last assertion follows from the fact that for positive $x\in \Me$, $\|x\|=\inf\{\lambda>0,x\le \lambda\}$.  

\end{proof}

To characterize arbitrary elements in $L_\infty(\Me,\varphi)$, let  $\Me_2:=M_2(\Me)$ be the algebra of $2\times 2$ matrices over $\Me$. The predual of $\Me_2$ can be identified with $M_2(\Me_*)$, where for $\psi\in (\Me_2)_*$, we put $\psi_{ij}(a)=\psi(a\otimes |i\>\<j|)$. This means that we also identify  $L_1(\Me_2)$  with $M_2(L_1(\Me))$.

\begin{lemma}\label{lemma:linfty} Let $k\in L_1(\Me)$. Let $h_2,k_2\in L_1(\Me_2)$ be defined as  
\[
h_2:=\left(\begin{array}{cc} h_\varphi & 0\\
                               0 & h_\varphi \end{array}\right),\quad 
k_2:=\left(\begin{array}{cc} 0 & k\\
                             k^*&  0 \end{array}\right).
\]
Then $k\in L_\infty(\Me, \varphi)$ if and only if $k_2\le \lambda h_2$ for some $\lambda>0$. In this  case,  
\[
\|k\|_{\infty,\varphi}=\inf\{\lambda>0, k_2\le\lambda h_2\}.
\]

\end{lemma}

\begin{proof} Let  $k=h_x$ and let $\lambda\in \mathbb R$.  Note that $\lambda h_2-k_2=
h_2^{1/2}x_\lambda h_2^{1/2}$, where 
\[
x_\lambda:=\left(\begin{array}{cc} \lambda & -x\\ -x^* & \lambda \end{array}\right),
\]
and that $\|x\|=\|-x\|=\inf\{\lambda>0,\  x_\lambda\ge 0\}$.
Hence $k_2\le \lambda h_2$ for any  $\lambda\ge \|x\|$.  It is also 
clear that $\|x\|$ is the smallest $\lambda$ such that this inequality holds. 

Conversely, assume that $k_2\le \lambda h_2$ for some $\lambda>0$, which is equivalent to $\Tr (\lambda h_2-k_2)a\ge 0$ for any $a\in \Me_2^+$, \cite[Thm. 33]{terp1981lpspaces}. Let $a=\left(\begin{array}{cc} a_{11} & a_{12}\\ a_{12}^* & a_{22}\end{array}\right)\in \Me_2^+$, then also 
\[
a_-:= \left(\begin{array}{cc} a_{11} & -a_{12}\\ -a_{12}^* & a_{22}\end{array}\right)
=\left(\begin{array}{cc} 1 &0  \\ 0 & -1\end{array}\right)a\left(\begin{array}{cc} 1 &0 \\  0 & -1\end{array}\right) \in \Me_2^+
\]
and note that $\Tr k_2a_-=- \Tr k_2a$, $\Tr h_2 a_-=\Tr h_2 a$.
 It follows that we have $\pm k_2\le \lambda h_2$, 
 so that $0\le k_2+\lambda h_2\le 2\lambda h_2$. 
Since $h_2$ defines a faithful positive normal linear functional on $\Me_2$, Lemma 
\ref{lemma:linftyplus} applies, so that there is some $y=\left(\begin{array}{cc} 
y_{11} & x \\ x^* &y_{22} \end{array}\right)\in \Me_2^+$ such that 
\[
\left(\begin{array}{cc} \lambda h_\varphi & k\\
     k^* & \lambda h_\varphi \end{array}\right)=
k_2+\lambda h_2=h_2^{1/2}yh_2^{1/2}=\left(\begin{array}{cc} h_\varphi^{1/2} y_{11} h_\varphi^{1/2} & h_\varphi^{1/2} x h_\varphi^{1/2}\\
h_\varphi^{1/2} x^* h_\varphi^{1/2} & h_\varphi^{1/2} y_{22} h_\varphi^{1/2}\end{array} \right).
\]
  It follows that   $y_{11}=y_{22}=\lambda$ and $k=h_x$. Moreover, since $y=\left(\begin{array}{cc} \lambda & x\\ x^* & \lambda \end{array}\right)$ is positive, $\|x\|\le \lambda$.

\end{proof}

\subsection{The interpolation spaces $L_p(\Me,\varphi)$}

We now define the $L_p$-space over $\Me$ with respect to $\varphi$ as
\[
L_p(\Me,\varphi):=C_{1/p}(L_\infty(\Me, \varphi),L_1(\Me)).
\]
For definition of the space $C_\theta$ see Appendix \ref{sec:interp}.

The norm in $L_p(\Me, \varphi)$ will be denoted by 
$\|\cdot\|_{p,\varphi}$. For $1\le p\le \infty$ and $1/q+1/p=1$, put
\[
i_p: L_p(\Me)\to L_1(\Me),\qquad k\mapsto h_\varphi^{1/2q}kh_\varphi^{1/2q}.
\]
\begin{thm}[{\cite[Theorem 9.1]{kosaki1984applications}}]\label{thm:lpfi} The map $i_p$ is an isometric isomorphism of $L_p(\Me)$ onto $L_p(\Me,\varphi)$.
\end{thm}

Using the polar decomposition  in $L_p(\Me)$ (\cite[Proposition 12]{terp1981lpspaces}), we obtain that elements in 
$L_p(\Me,\varphi)$ have the form $h_\varphi^{1/2q}u h_\psi^{1/p} h_\varphi^{1/2q}$, where 
$\psi\in \Me_*^+$ and $u\in \Me$ is a partial isometry such that $u^*u=s(\psi)$
with norm 
\[
\|h^{1/2q}_\varphi uh_\psi^{1/p} h^{1/2q}_\varphi\|_{p,\varphi}=(\Tr h_\psi)^{1/p}=\psi(1)^{1/p}.
\]

\begin{ex}\label{ex:semifinite}
Assume that $\Me$ is semifinite and let $\tau$ be a faithful normal semifinite trace on $\Me$. By \cite[p. 62]{terp1981lpspaces}, $L_p(\Me)$ can be identified with 
the space $L_p(\tau)$ of closed densely defined operators $X$ affiliated with $\Me$, such that $\tau(|X|^p)<\infty$, with the norm $\|X\|_p=\tau(|X|^p)^{1/p}$.
There is an operator $\rho_\varphi\in L_1(\tau)$ such that 
\[
\varphi(x)=\tau(\rho_\varphi x), \qquad x\in \Me
\]
 and we can define the embedding $\Me\subseteq L_1(\tau)\equiv \Me_*$
 as $x\mapsto \rho_\varphi^{1/2}x\rho_\varphi^{1/2}$. The space $L_p(\Me,\varphi)$ can be identified with the subspace in $L_1(\tau)$ of elements of the form 
$X=\rho_\varphi^{1/2q}Y\rho_\varphi^{1/2q}$ with $Y\in L_p(\tau)$, and $\|X\|_{p,\varphi}=\|Y\|_p$. 
In particular, if $\Me$ is finite dimensional, then $L_p(\Me,\varphi)\equiv L_p(\Me)\equiv \Me$ as linear spaces and we have
\[
\|X\|_{p,\varphi}=\|\rho_\varphi^{-1/2q}X\rho_\varphi^{-1/2q}\|_p=(\Tr| \rho_\varphi^{\frac{1-p}{2p}}X\rho_\varphi^{\frac{1-p}{2p}}|^p)^{1/p}.
\]

\end{ex}

We now list some important properties of the spaces $L_p(\Me,\varphi)$.
Let $1\le p\le p'\le \infty$. Then $L_{p'}(\Me,\varphi)\subseteq L_p(\Me,\varphi)$ and 
 \begin{equation}\label{eq:embedding}
 \|k\|_{p,\varphi}\le \|k\|_{p',\varphi},\qquad \forall k\in L_{p'}(\Me,\varphi).
 \end{equation}
 This follows easily by Theorem \ref{thm:lpfi} and  H\"older's inequality, 
but it is also a consequence of the abstract theory of complex interpolation, see 
\cite[Theorem 4.2.1]{belo1976interpolation}. The space $L_\infty(\Me, \varphi)$ is dense in $L_1(\Me)$   and therefore also in $L_p(\Me,\varphi)$ for each $p>1$ by \cite[Theorem 4.2.2]{belo1976interpolation}. It follows that $L_{p'}(\Me,\varphi)$ and $L_p(\Me,\varphi)$ are compatible Banach spaces. By the reiteration theorem 
(\cite[Theorem 4.6.1]{belo1976interpolation}), we have 
\begin{equation}\label{eq:reiteration}
 C_\eta(L_{p'}(\Me,\varphi),L_p(\Me,\varphi))=L_{p_\eta}(\Me,\varphi),\qquad 0\le \eta\le 1,
\end{equation}
where $1/p_\eta= \eta/p+(1-\eta)/p'$.

Let now $1\le p\le \infty$, $1/p+1/q=1$. The duality
\[
\<k,h_x\>:= \Tr kx, \qquad x\in \Me,\ k\in L_1(\Me)
\]
extends to a duality between 
$L_p(\Me,\varphi)$ and $L_q(\Me,\varphi)$,  given by
\begin{equation}\label{eq:duality}
\<h_\varphi^{1/2q}k_1h_\varphi^{1/2q},h_\varphi^{1/2p}k_2h_\varphi^{1/2p}\>=\Tr k_1k_2,\quad  k_1\in L_p(\Me), k_2\in L_q(\Me).
\end{equation}
 For $1\le p<\infty$, $L_q(\Me,\varphi)$ is isometrically isomorphic to the Banach space dual of $L_p(\Me,\varphi)$. This follows immediately from Theorem \ref{thm:lpfi}
and duality of Haagerup $L_p$-spaces \cite[Thm. 32]{terp1981lpspaces}.

 For each $1\le p\le\infty$,  we have the following Clarkson type  inequalities. 

\begin{thm}[{\cite{kosaki1984applications},\cite[Thm. 5.1]{pixu2003noncommutative}}] \label{thm:clarkson} Let $h,k\in L_p(\Me,\varphi)$, $1\le p\le \infty$, $1/p+1/q=1$. 
 For $2\le p\le\infty$ we have 
\[
\left[ \frac12(\|h+k\|_{p,\varphi}^p+\|h-k\|_{p,\varphi}^p\right]^{1/p}\le \left(\|h\|_{p,\varphi}^q+\|k\|_{p,\varphi}^q\right)^{1/q}.
\]
For $1\le p\le 2$ the inequality reverses. 
\end{thm}

This implies that for $1<p<\infty$ the space $L_p(\Me,\varphi)$ is uniformly convex and uniformly smooth. 
We also have:   

\begin{thm}[{\cite[Thm 5.3]{pixu2003noncommutative}}] \label{thm:pixu} Let $h,k\in L_p(\Me,\varphi)$, $1< p\le 2$, then
\[
\left( \|h\|_p^2+(p-1)\|k\|_p^2\right)^{1/2}\le \left[\frac12( \|h+k\|_p^p+\|h-k\|_p^p)\right]^{1/p}.
\]
For $2< p<\infty$ the inequality reverses.
\end{thm}

The space  $L_p(\Me,\varphi)$ is strictly convex, hence for each $0\ne h\in L_p(\Me,\varphi)$, there is a unique element $T_{q,\varphi}(h)$ in the unit ball of $L_q(\Me,\varphi)$
 such that \[
\<T_{q,\varphi}(h),h\>=\|h\|_{p,\varphi}.
\]
Let $h=h_\varphi^{1/2q}uk^{1/p}h_\varphi^{1/2q}$ for some $k\in L_1(\Me)^+$ and a 
partial isometry $u\in \Me$ such that $u^*u=s(k)$.  Then by \eqref{eq:duality} we  have 
\begin{equation}\label{eq:dual}
T_{q,\varphi}(h)=\|h\|_{p,\varphi}^{1-p}h_\varphi^{1/2p}k^{1/q}u^*h_\varphi^{1/2p}.
\end{equation}
Restricted to the unit sphere of $L_p(\Me,\varphi)$, the map $T_{q,\varphi}$ is a uniformly continuous bijection onto the unit sphere of $L_q(\Me,\varphi)$ \cite{diestel1975geometry} and we have 
$T_{q,\varphi}^{-1}=T_{p,\varphi}$ for this restriction.

\subsection{Hadamard three lines theorem}\label{sec:hadamard}

We first note that  the infimum in the definition of the interpolation norm $\|\cdot\|_{p,\varphi}$   is attained, see \eqref{eq:theta}.
Let $h\in L_p(\Me, \varphi)$ be of the form $h=h_\varphi^{1/2q}kh_\varphi^{1/2q}$ for some $k\in L_p(\Me)$ and let 
$k=ul^{1/p}$ be the polar decomposition of $k$. Let $S\subset \mathbb C$ be the strip $S=\{z\in \mathbb C,\ 0\le Re(z)\le 1\}$ and put
\begin{equation}\label{eq:fkp}
f_{h,p}(z):=\|l\|_1^{1/p-z}h_\varphi^{(1-z)/2}ul^{z}h_{\varphi}^{(1-z)/2},\qquad z\in S.
\end{equation}
Then $f_{h,p}\in \Fe:=\Fe(L_\infty(\Me,\varphi), L_1(\Me))$, $f_{h,p}(1/p)=h$  and we have $\|h\|_{p,\varphi}=\vertiii{f_{h,p}}_\Fe$, cf. \cite[proof of Theorem 9.1]{kosaki1984applications}, see Appendix \ref{sec:interp} for the necessary definitions.

\begin{lemma}\label{lemma:equal} Let $f\in \Fe$ and assume that  $\|f(\theta)\|_{1/\theta,\varphi}=\vertiii{f}_\Fe$ for some $\theta\in(0,1)$. Then 
\[
\|f(x+it)\|_{1/x,\varphi}=\vertiii{f}_\Fe,\qquad \forall x\in [0,1],\ t\in \mathbb R.
\]

\end{lemma}

\begin{proof} Let $p=1/\theta$, $q=1/(1-\theta)$. Put $h:=f(\theta)$, then $h\in L_p(\Me,\varphi)$ and $g:=  
f_{T_{q,\varphi}(h),q}$ is in $\Fe$. Let
\[
K(z):=\<g(1-z),f(z)\>,\qquad z\in S.
\]
Note that for $z=x+it$, $f(z)\in L_{1/x}(\varphi)$, $g(1-z)\in L_{1/(1-x)}(\varphi)$ and 
$\|f(z)\|_{1/x,\varphi}\le \vertiii{f}_\Fe$, $\|g(1-z)\|_{1/(1-x),\varphi}\le \vertiii{g}_\Fe=1$. It follows that $K$ is continuous on $S$, analytic in the interior and bounded by
\[
|K(x+it)|\le \|g(1-x-it)\|_{1/(1-x),\varphi}\|f(x+it)\|_{1/x,\varphi}\le \vertiii{f}_\Fe.
\]
Moreover, $K(\theta)=\|f(\theta)\|_{p,\varphi}=\vertiii{f}_\Fe$. By the maximum modulus principle, $K$ must be a constant, so that 
$K(z)=\vertiii{f}_\Fe$ for all $z\in S$. It follows that we must have $\|f(x+it)\|_{1/x,\varphi}= \vertiii{f}_\Fe$ for all $x$ and $t$.

\end{proof}

The next lemma shows that the infimum in the definition of the interpolation norm is attained also for the reiterated spaces.

\begin{lemma} \label{lemma:reiteration}
Let $1\le p\le p'\le \infty$ and let $\eta\in (0,1)$, $p_\eta= \eta/p+(1-\eta)/p'$. Let $h\in L_{p_\eta}(\Me,\varphi)$ and put 
$g(z)=f_{h,p_\eta}(z/p+(1-z)/p')$, $z\in S$. Then $g\in \Fe_{p',p}:=\Fe(L_{p'}(\Me,\varphi), L_p(\Me,\varphi))$, $g(\eta)=h$ and 
$\|h\|_{p_\eta,\varphi}=\vertiii{g}_{\Fe_{p',p}}$.
\end{lemma}

\begin{proof}
By \cite[32.3]{calderon1964intermediate}, for any  $f\in \Fe$ the function 
$q(z)=f(z/p+(1-z)/p')$ belongs to 
$\Fe_{p',p}$ and 
\[
\vertiii{q}_{\Fe_{p',p}}=\max\{\sup_t\|q(it)\|_{p',\varphi},\sup_t\|q(1+it)\|_{p,\varphi}\}\le \vertiii{f}_\Fe.\]
By reiteration \eqref{eq:reiteration},
\[
\|h\|_{p_\eta,\varphi}\le \vertiii{g}_{\Fe_{p',p}}\le \vertiii{
f_{h,p_\eta}}_\Fe= \|h\|_{p_\eta,\varphi}.
\]
The statement follows also by Lemma \ref{lemma:equal}, by noticing that 
for any $x\in [0,1]$, $t\in \mathbb R$, 
\[
\|g(x+it)\|_{p_x,\varphi}=\vertiii{f_{h,p_\eta}}_\Fe=\|h\|_{p_\eta,\varphi}.
\]

\end{proof}

Assume that $h\in L_p(\Me,\varphi)$, $\|h\|_{p,\varphi}=1$. Note that by Lemma \ref{lemma:equal}, the values of the 
function $f_{h,p}$ run through the unit spheres  of all the spaces $L_{p'}(\Me,\varphi)$. The next lemma shows that by applying the map $T_{q,\varphi}$ we again obtain an element of $\Fe$.

\begin{lemma}\label{lemma:Tq} Let $1<p<\infty$, and let $h\in L_p(\Me,\varphi)$, with 
$\|h\|_{p,\varphi}=1$.   Then for all $z=x+it$, 
$x\in (0,1)$, 
\[
T_{1/(1-x),\varphi}(f_{h,p}(z))=f_{T_{q,\varphi}(h),q}(1-z).
\]

\end{lemma}

\begin{proof} Since $\|h\|_{p,\varphi}=1$, we have  
 $h=h_\varphi^{1/2q}uh_\psi^{1/p}h_\varphi^{1/2q}$ for some 
 $\psi\in \states_*(\Me)$ and $u^*u=s(\psi)$. 
By Lemma \ref{lemma:equal}, $\|f_{h,p}(x+it)\|_{1/x,\varphi}=\vertiii{f_{h,p}}_\Fe=1$
and $\|f_{T_{q,\varphi}(h),q}(1-x-it)\|_{1/(1-x),\varphi}=\vertiii{f_{T_{q,\varphi}(h),q}}_\Fe=1$
 for all $x\in [0,1]$ and $t\in \mathbb R$. 
 By \eqref{eq:duality}, we have
\begin{align*}
\<f_{h,p}(z), f_{T_{q,\varphi}(h),q}(1-z)\>&= \Tr \left(h^{-it/2}_\varphi uh_\psi^z h_\varphi^{-it/2}\right)\left(h_\varphi^{it/2}h_\psi^{1-z}u^*h_\varphi^{it/2}\right)\\
&=\Tr h_\psi=1.
\end{align*}
By uniqueness, we must have $T_{1/(1-x),\varphi}(f_{h,p}(z))=f_{T_{q,\varphi}(h),q}(1-z)$ for all $x\in (0,1)$, $t\in \mathbb R$. 
 
\end{proof}

The inequality part of the following result  is a version of  Hadamard's three lines theorem. For  convenience of the reader, we  add a proof.

\begin{thm}\label{thm:htl} Let $1\le p\le p'\le\infty$ and let $0<\eta<1$. Then for $f\in \Fe(L_{p'}(\Me,\varphi),L_p(\Me, \varphi))$,
\[
\|f(\eta)\|_{p_\eta,\varphi}\le (\sup_{t\in\mathbb R}\|f(it)\|_{p',\varphi})^{1-\eta}(\sup_{t\in\mathbb R}\|f(1+it)\|_{p,\varphi})^{\eta},
\]
where $1/p_\eta=\eta/p+(1-\eta)/p'$. Moreover, equality is attained if and only if
\[
f(z)=f_{h,p_\eta}(z/p+(1-z)/p')M^{z-\eta},
\]
 for $h=f(\eta)$ and  $M>0$. 

\end{thm}

\begin{proof} Let $q$, $q'$ and $q_\eta$ be the duals of $p$, $p'$ and $p_\eta$, so that $1/q_{\eta}=\eta/q+(1-\eta)/q'$ and 
$L_{q_\eta}(\Me, \varphi)=C_{1-\eta}(L_q(\Me,\varphi),L_{q'}(\Me,\varphi))$. Put  
$h=f(\eta)$ and let 
\[
g(z)=f_{T_{q_\eta,\varphi} (h),q_\eta}(z/q'+(1-z)/q). 
\]
By Lemma \ref{lemma:reiteration}, $g  \in \Fe_{q,q'}$ 
and \[
g(1-\eta)=T_{q_\eta,\varphi}(h),\qquad 
\vertiii{g}_{\Fe_{q,q'}}=1=\|T_{q_\eta,\varphi}(h)\|_{q_\eta,\varphi}.
\]
As in the proof of lemma \ref{lemma:equal}, $K(z):=\<g(1-z),f(z)\>$ defines a bounded continuous function on $S$, analytic in the interior of $S$.
By the usual Hadamard's three lines theorem,
\begin{align*}
\|f(\eta)\|_{p_\eta,\varphi}&=|K(\eta)|\le (\sup_{t\in\mathbb R}|K(it)|)^{1-\eta}(\sup_{t\in\mathbb R}|K(1+it)|)^{\eta}\\
&\le 
(\sup_{t\in\mathbb R}\|f(it)\|_{p',\varphi})^{1-\eta}(\sup_{t\in\mathbb R}\|f(1+it)\|_{p,\varphi})^{\eta}.
\end{align*}
Now assume that equality is attained. Let $M_0=\sup_{t\in\mathbb R}\|f(it)\|_{p',\varphi}$, $M_1=\sup_{t\in\mathbb R}\|f(1+it)\|_{p,\varphi}$ and let
\[
F(z):=K(z)M_0^{z-1}M_1^{-z},\qquad z\in S.
\]
Then $|F(z)|\le 1$ for all $z\in S$ and $F(\eta)=1$. By the maximum modulus principle, $F(z)=F(\eta)=1$ for all $z$, that is,
\begin{equation}\label{eq:hadeq}
\<g(1-z),f(z)M_0^{z-1}M_1^{-z}\>=1,\qquad z\in S.
\end{equation}
Suppose first that $\|h\|_{p_\eta,\varphi}=1$. Note that by Lemma \ref{lemma:Tq}, 
\[
g(1-z)=f_{T_{q_\eta,\varphi}(h),q_\eta}(u)=T_{1/Re(u),\varphi}(f_{h,p_\eta}(1-u)),
\]
where $u=z/q+(1-z)/q'$. Hence
\[
g(1-z)=T_{q_x,\varphi}(f_{h,p_\eta}(z/p+(1-z)/p')).
\]
Since $\|f(x+it)M_0^{x+it-1}M_1^{-x-it}\|_{p_x,\varphi}\le 1$ by the first part of the 
proof, \eqref{eq:hadeq} implies that we must have 
\[
f(z)M_0^{z-1}M_1^{-z}=f_{h,p_\eta}(z/p+(1-z)/p'),
\]
by definition and properties  of $T_{q_x,\varphi}$.
If  $\|h\|_{p_\eta,\varphi}=a\ne 1$, then we may replace $f$ by $a^{-1}f$. Note that the above equality still holds, with $M_0$ and $M_1$ replaced by $a^{-1}M_0$ and $a^{-1}M_1$. We obtain
\begin{align*}
f(z)&=a(a^{-1}f(z))=af_{a^{-1}h,p_\eta}(z/p+(1-z)/p')(a^{-1}M_0)^{1-z}(a^{-1}M_1)^z
\\
&=f_{a^{-1}h,p_\eta}(z/p+(1-z)/p')M_0^{1-z}M_1^z\\
&=a^{-1/p_\eta}f_{h,p_\eta}(z/p+(1-z)/p')M_0^{1-z}M_1^z\\
&= f_{h,p_\eta}(z/p+(1-z)/p')AM^z,
\end{align*}
where $A>0$ and $M=M_1/M_0$. Since $f(\eta)=f_{h,p_\eta}(1/p_\eta)=h$, we must have 
$AM^\eta=1$. It follows that 
\[
f(z)=f_{h,p_\eta}(z/p+(1-z)/p')M^{z-\eta}.
\]

For the converse, note that using Lemma \ref{lemma:equal}, we obtain  
\begin{align*}
M_0&=\sup_t\|f(it)\|_{p',\varphi}=\|h\|_{p_\eta,\varphi}M^{-\eta}\\
M_1&=\sup_t\|f(1+it)\|_{p,\varphi}=\|h\|_{p_\eta,\varphi}M^{1-\eta}.
\end{align*}
It follows that $M_1/M_0=M$ and $M_0^{1-\eta}M_1^\eta=\|h\|_{p_\eta,\varphi}$.
 
\end{proof}

\subsection{The positive cone in $L_p(\Me,\varphi)$.}

Let us denote $L_p(\Me,\varphi)^+:=L_p(\Me,\varphi)\cap L_1(\Me)^+$. Then it is clear that for $1<p<\infty$,
\[
L_p(\Me,\varphi)^+=
\{h_\varphi^{1/2q}h^{1/p}h_\varphi^{1/2q}, \ h\in L_1(\Me)^+\}.
\]
It follows by the properties of $L_p(\Me)^+$ (\cite{terp1981lpspaces})  that $L_p(\Me,\varphi)^+$ is a closed convex cone which is pointed and generates all $L_p(\Me,\varphi)$. 
Note also that 
\[
L_\infty(\Me,\varphi)^+=\{h_\varphi^{1/2}xh_\varphi^{1/2}, \ x\in\Me^+\}
\]
 is dense in $L_p(\Me,\varphi)^+$, for any $1\le p$.

Let $1<p<\infty$ and let $k\in L_p(\Me,\varphi)$,  $k=h_\varphi^{1/2q}u h^{1/p}h_\varphi^{1/2q}$, $h\in L_1(\Me)^+$. Then $k$  has a polar decomposition of the form
\[
k=h_\varphi^{1/2q}uh_\varphi^{-1/2q}|k|_{p,\varphi}=\sigma^\varphi_{-i/2q}(u)|k|_{p,\varphi},
\]
where $|k|_{p,\varphi}=h_\varphi^{1/2q}h^{1/p}h_\varphi^{1/2q}\in L_p(\Me,\varphi)^+$ and 
$\sigma^\varphi$ denotes the modular group of $\varphi$.  
We next look at  self-adjoint elements in $L_p(\Me,\varphi)$. 

\begin{lemma}\label{lemma:jdecomp} Let $1\le p<\infty$ and $k=k^*\in L_p(\Me,\varphi)$. Then there is a decomposition 
\[
k=k_{p,\varphi,+}-k_{p,\varphi,-},
\]
where $k_{p,\varphi,\pm}\in L_p(\Me,\varphi)^+$ and we have 
\[
\|k\|_{p,\varphi}=(\|k_{p,\varphi,+}\|_{p,\varphi}^p+\|k_{p,\varphi,-}\|_{p,\varphi}^p)^{1/p}.
\]

\end{lemma}

\begin{proof}
If $k=k^*$, then $k=h_\varphi^{1/2q}lh_\varphi^{1/2q}$, where $l=l^*\in L_p(\Me)$. It follows that $l=u|h|^{1/p}$, where 
$h=h^*\in L_1(\Me)$, $h=h_+-h_-$, $h_+,h_-\in L_1(\Me)^+$, $h_+h_-=0$. Moreover, $u=e_+-e_-$, where $e_\pm:=s(h_\pm)$ and $|h|^{1/p}=(h_++h_-)^{1/p}=h_+^{1/p}+h_-^{1/p}$. It follows that $k$ has the above form, with    $k_{p,\varphi,\pm}=    h_\varphi^{1/2q}h_\pm^{1/p}h_\varphi^{1/2q}$ and we have 
\[
\|k\|_{p,\varphi}^p= \Tr |h| = \Tr h_++\Tr h_-=\|k_{p,\varphi,+}\|_{p,\varphi}^p+\|k_{p,\varphi,-}\|_{p,\varphi}^p.
\]

\end{proof}

\begin{coro}\label{coro:positive} Let $h\in L_p(\Me,\varphi)^+$ and let $h_1 \in L_1(\Me)^+$ be such that $h_1\le h$. Then 
$h_1\in L_p(\Me,\varphi)^+$ and $\|h_1\|_{p,\varphi}\le \|h\|_{p,\varphi}$.

\end{coro}

\begin{proof} The statement is obvious for $p=1$ and follows easily from Lemma \ref{lemma:linftyplus} for $p=\infty$. For $1<p<\infty$, let $x\in \Me^+$, then  
\[
0\le \<h_x,h_1\>=\Tr h_1x \le \Tr hx =\<h_x,h\>\le \|h_x\|_{q,\varphi}\|h\|_{p,\varphi}.
\]
Since $L_\infty(\Me,\varphi)^+$ is dense in $L_q(\Me,\varphi)^+$, it follows that
$\<k,h_1\>\le \<k,h\>\le \|k\|_{q,\varphi}\|h\|_{p,\varphi}$ for all $k\in L_q(\Me,\varphi)^+$. 
Let now $k=k^*\in L_q(\Me,\varphi)$, with decomposition  $k=k_{q,\varphi,+}-k_{q,\varphi,-}$ as in Lemma \ref{lemma:jdecomp}. Then
\begin{align*}
|\<k,h_1\>|&\le \<k_{q,\varphi,+},h_1\>+\<k_{q,\varphi,-},h_1\>\le \|h\|_{p,\varphi}(\|k_{q,\varphi,+}\|_{q,\varphi}+\|k_{q,\varphi,-}\|_{q,\varphi})\\
&\le  \|h\|_{p,\varphi}2^{1/q}(\|k_{q,\varphi,+}\|^q_{q,\varphi}+\|k_{q,\varphi,-}\|^q_{q,\varphi})^{1/q}= \|h\|_{p,\varphi}2^{1/q}\|k\|_{q,\varphi},
\end{align*}
the last inequality follows by classical H\"older's inequality. For $k\in L_q(\Me,\varphi)$, we have $k=\mathrm {Re}(k)+ i\mathrm{Im}(k)$, with the usual definition of the self-adjoint elements $\mathrm {Re}(k)$ and $\mathrm{Im}(k)$ in $L_q(\Me,\varphi)$. Then
\begin{align*}
|\<k,h_1\>|&\le |\<\mathrm{Re}(k),h_1\>|+|\<\mathrm{Im}(k),h_1\>|
\\
&\le \|h\|_{p,\varphi}2^{1/q}(\|\mathrm{Re}(k)\|_{q,\varphi}+\|\mathrm{Im}(k)\|_{q,\varphi})\le  \|h\|_{p,\varphi}2^{1+1/q}\|k\|_{q,\varphi}.
\end{align*}
Hence $h_1$ defines a bounded positive linear functional on $L_q(\Me,\varphi)$ and therefore 
$h_1\in L_p(\Me,\varphi)^+$. To prove the last statement, note that by (\ref{eq:dual}), 
$T_{q,\varphi}(h_1)$ is a positive element in the unit ball of $L_q(\Me,\varphi)$, so that
\[
\|h_1\|_{p,\varphi}=\<T_{q,\varphi}(h_1),h_1\>\le \<T_{q,\varphi}(h_1),h\>\le \|h\|_{p,\varphi}.
\]

\end{proof}

\section{The  R\'enyi relative entropy}\label{sec:sandwiched}

We will need to extend the definition of $L_p(\Me,\varphi)$ to all (not necessarily faithful) normal states.
 So let $\varphi\in \states_*(\Me)$ and let $s(\varphi)=e$. Then $\varphi$ restricts to a faithful normal state on $e\Me e$ and we may identify the predual $(e\Me e)_*$ with the set of all $\psi\in \Me_*$ such that $e\psi e=\psi$, where 
$e\psi e(x)=\psi(exe)$, $x\in \Me$. By \cite[Theorem 7]{terp1981lpspaces}, $h_\psi=h_{e\psi e}=eh_\psi e$ for all such $\psi$. Hence we may identify $L_1(e\Me e)\simeq eL_1(\Me) e$ and using the polar decomposition, 
$L_p(e\Me e)\simeq eL_p(\Me) e$ for all $p\ge 1$. The space $L_p(\Me, \varphi)$ is then defined as 
\[
L_p(\Me,\varphi)= \{h\in L_1(\Me),\ h=ehe\in L_p(e\Me e,\varphi|_{e\Me e})\}, 
\]
with the corresponding norm.

\subsection{Definition and basic properties}

Let $1<\alpha< \infty$ and let $\varphi$, $\psi\in \states_*(\Me)$. We define

\begin{equation}\label{eq:defi_sre}
\tilde D_\alpha(\psi\|\varphi)=\left\{\begin{array}{cc} \frac{\alpha}{\alpha-1}\log(\|h_\psi\|_{\alpha,\varphi}) & \mbox{if }  h_\psi\in L_\alpha(\Me,\varphi)\\ & \\
\infty & \mbox{otherwise}.
\end{array}\right.
\end{equation}

We first show that this definition is an extension of the sandwiched R\'enyi relative $\alpha$-entropy \eqref{eq:sandwiched}.  
Assume that
$\dim(\Me)<\infty$ and let $\tau_0$ be a faithful normal trace on $\Me$. Any state $\varphi\in \states_*(\Me)$ is given by  a density operator $\rho_\varphi\in \Me^+$, such that 
$\varphi(x)=\tau_0 (\rho_\varphi x)$, $x\in \Me$. 

\begin{prop}\label{prop:finitedim} Let $\dim(\Me)<\infty$ and let $\psi,\varphi\in \states_*(\Me)$, with density operators 
$\rho_\varphi=\sigma, \rho_\psi=\rho$. Then
\[
\tilde D_\alpha(\psi\|\varphi)=\left\{\begin{array}{cc}\frac1{\alpha-1}\log \tau_0 [(\sigma^{\frac{1-\alpha}{2\alpha}}\rho\sigma^{\frac{1-\alpha}{2\alpha}})^\alpha], & \mbox{ if } \supp(\rho)\subseteq  \supp(\sigma)\\ & \\
\infty & \mbox{otherwise}.
\end{array}\right.
\]

\end{prop}

\begin{proof} We may assume that $s(\psi)\le s(\varphi)=:e$, otherwise $\psi\notin L_p(\Me,\varphi)$ and $\supp(\rho)\not\subseteq  \supp(\sigma)$, hence both quantities are infinite. It follows that $\sigma=e\sigma e$, $\rho=e\rho e$. The statement now follows by Example \ref{ex:semifinite}.

\end{proof}

\subsubsection{Relation to the Araki-Masuda divergences}

In this paragraph, we show that $\tilde D_\alpha$ is equal to the 
Araki-Masuda divergences  introduced in \cite{berta2018renyi}. These divergences are based on  the Araki-Masuda definition of the noncommutative $L_p$-spaces, \cite{arma1982positive}. 

The Araki-Masuda $L_p$-spaces are defined with respect to a faithful state $\varphi\in \states_*(\Me)$, using a standard form of $\Me$ and a vector representative $\eta$
of $\varphi$. 
We will use the standard form $(\lambda(\Me), L_2(\Me), J, L_2(\Me)^+)$ and $\eta=h_\varphi^{1/2}$. For $\xi\in L_2(\Me)$ and $2\le p<\infty$, the Araki-Masuda $L_p$-norm is 
defined as  \cite[Eq. (1.4)]{arma1982positive}
\[
\|\xi\|_{p,\varphi}^{AM}=\sup_{\zeta\in L_2(\Me), \|\zeta\|=1}\|\Delta_{\zeta,\eta}^{1/2-1/p}\xi\|,
\]
where $\Delta_{\zeta,\eta}$ is the relative modular operator, see Appendix \ref{sec:rmo}. The Araki-Masuda $L_p$ space is then the subspace of 
 vectors  $\xi \in L_2(\Me)$ such that this expression is finite. By the polar decomposition \cite[Thm. 3]{arma1982positive} and \eqref{eq:app_modular}, 
this happens if and only if
\begin{equation}\label{eq:polar_AM}
\xi=uh_\mu^{1/p}h_\varphi^{1/2-1/p}
\end{equation}
for some (unique) partial isometry $u\in \Me$ and $\mu\in \Me_*^+$, such that $uu^*=s(\omega_\xi)$ and $u^*u=s(\mu)$. In this case, 
$\|\xi\|_{p,\varphi}^{AM}=\mu(1)^{1/p}$. 

The weighted $L_p$-norm of \cite{berta2018renyi} is defined for any *-representation  $\pi: \Me\to B(\mathcal H)$ using the spatial derivative. By restriction to the support, the definition can be extended to non-faithful states. For $\varphi\in \states_*(\Me)$, $2\le p\le \infty$  and $\xi\in \Ha$, put
\[
\|\xi\|_{p,\varphi}^{BST}:= \left\{\begin{array}{cc} \sup_{\zeta\in \mathcal H, \|\zeta\|=1}\|\Delta(\zeta/\varphi)^{1/2-1/p}\xi\|,& \mbox{ if } s(\omega_\xi)\le s(\varphi)\\
      \infty & \mbox{otherwise}\end{array}\right.,
\]
here $\Delta(\zeta/\varphi)$ is the spatial derivative (see Appendix \ref{sec:spatial} for the definition.) 

The Araki-Masuda divergence is defined as follows. Let $\varphi,\psi\in \states_*(\Me)$ and $1<\alpha<\infty$. Let $\pi:\Me\to B(\mathcal H)$ be any *-representation and let $\xi_\psi\in \mathcal H$ be any vector representative of $\psi$. Then 
\begin{equation}\label{eq:AM_div}
D^{AM}_\alpha(\psi\|\varphi):= \frac{2\alpha}{\alpha-1}\log\|\xi_\psi\|_{2\alpha,\varphi}^{BST}.
\end{equation}
By \cite[Lemma 3]{berta2018renyi}, the value of $\|\xi\|_{p,\varphi}^{BST}$ depends only on the functional $\omega_\xi$, not on the representation $\pi$ or the representing 
vector $\xi$. Therefore, $D^{AM}_\alpha$ is well defined. Moreover, we may use the above standard form. By the equality \eqref{eq:app_spatial}, we see that if $\varphi$ is faithful, we have
\[
\|\xi\|_{p,\varphi}^{BST}=\|\xi^*\|_{p,\varphi}^{AM},\qquad \forall \xi\in L_2(\Me). 
\]
By  the polar decomposition \eqref{eq:polar_AM}, we obtain obtain the following statement.

\begin{lemma}\label{lemma:polar_BST}
 Let $\xi\in L_2(\Me)$ and let $\varphi$ be faithful. Then $\|\xi\|_{p,\varphi}^{BST}<\infty$ 
if and only if $\xi= h_\varphi^{1/2-1/p}h_\mu^{1/p}v$ for some partial isometry $v\in \Me$ and $\mu\in \Me_*^+$, satisfying $v^*v=s(\omega_\xi)$ and $vv^*=s(\mu)$.
Moreover, such $v$ and $\mu$ are unique and we have  $\|\xi\|_{p,\varphi}^{BST}=\mu(1)^{1/p}$.
\end{lemma}

We can now prove the main result of this paragraph.

\begin{thm}\label{thm:AM} For any $\varphi,\psi\in \states_*(\Me)$ and $1<\alpha<\infty$, $D^{AM}_\alpha(\psi\|\varphi)=\tilde D_\alpha(\psi\|\varphi)$.

\end{thm}

\begin{proof} 
We may assume that $s(\psi)\le s(\varphi)$, otherwise both expressions are infinite. By restriction to the compressed algebra $s(\varphi)\Me s(\varphi)$, 
we may suppose that $\varphi$ is faithful. 

Assume that $D^{AM}_\alpha(\psi\|\varphi)$ is finite. Let $\xi_\psi\in L_2(\Me)$ be any vector representative of $\psi$, by \eqref{eq:Haage_innerp}
 this is equivalent to  $\xi_\psi\xi_\psi^*=h_\psi$. 
By Lemma \ref{lemma:polar_BST}
\[
\xi_\psi=h_\varphi^{1/2-1/2\alpha}h_\mu^{1/2\alpha}u=h_\varphi^{1/2\beta}h_\mu^{1/2\alpha}u
\]
for some $\mu\in \Me_*^+$ and a partial isometry $u$ with $u^*u=s(\psi)$, $uu^*=s(\mu)$. Then $h_\psi=\xi_\psi\xi_\psi^*=h_\varphi^{1/2\beta}h_\mu^\alpha h_\varphi^{1/2\beta}$, so that $h_\psi\in L_\alpha(\Me,\varphi)$ and 
\[
\|h_\psi\|_{\alpha,\varphi}=\mu(1)^{1/\alpha}=(\|\xi_\psi\|_{2\alpha,\varphi}^{BST})^2. 
\]
It follows that 
$D^{AM}_\alpha(\psi\|\varphi)=\tilde D_\alpha(\psi\|\varphi)$.
If $\tilde D_\alpha(\psi\|\varphi)<\infty$, then $h_\psi\in L_p(\Me,\varphi)$, so that $h_\psi=h_\varphi^{1/2\beta}h_\mu^{1/\alpha}h_\varphi^{1/2\beta}$ for some
 $\mu\in \Me_*^+$. Put $\xi:=h_\varphi^{1/2\beta}h_\mu^{1/2\alpha}$, then $\xi\in L_2(\Me)$ is a vector representative of $\psi$. Using again Lemma \ref{lemma:polar_BST},
 we have $\|\xi\|_{2\alpha,\varphi}^{BST}=\mu(1)^{1/2\alpha}=\|h_\psi\|_{\alpha,\varphi}^{1/2}$, this implies the result.

\end{proof}

\begin{rem} Using \cite[Theorem 9.1]{kosaki1984applications}, the  mapping $L_2(\Me)\ni k\mapsto h_\varphi^{1/2}k\in L_1(\Me)$ 
and Lemma \ref{lemma:polar_BST}, we can see that the BST-norms can be obtained by complex interpolation as follows. Consider the 
map
\[
 x\mapsto h_\varphi^{1/2}x,\qquad x\in \Me.
\]
By H\"older inequality, it is a continuous embedding of $\Me$ into $L_2(\Me)$. The norm $\|\cdot\|^{BST}_{p,\varphi}$ is 
then the norm of the interpolation space $C_{1/p}(\Me,L_2(\Me))$. 

\end{rem}

\subsubsection{Relation to standard R\'enyi relative entropies}

Let $\psi,\varphi\in \states_*(\Me)$. The standard version of the R\'enyi relative entropy $D_\alpha$ for $\alpha>0$, $\alpha\ne 1$ 
was defined by Petz  \cite{petz1985quasi, ohpe1993entropy} and can be written using the relative modular operator $\Delta_{\psi,\varphi}$
(see  Appendix \ref{sec:rmo}):
\[
D_\alpha(\psi\|\varphi)=\left\{\begin{array}{cc} \frac 1{\alpha-1} \log (h_\varphi^{1/2},\Delta_{\psi,\varphi}^{\alpha} h_\varphi^{1/2}),  & \mbox{ if } s(\psi)\le s(\varphi)\\ & \\
\infty & \mbox{otherwise},
\end{array}\right.
\]
 Note that by \eqref{eq:app_modular} we may formally write 
\[
D_\alpha(\psi\|\varphi)=\frac 1{\alpha-1} \log(\Tr h_\psi^\alpha h_\varphi^{1-\alpha}).
\]
If $\alpha\in(0,1)$, this quantity is always well defined and finite, the function $\alpha\mapsto D_\alpha(\psi\|\varphi)$ is increasing and the limit for $\alpha\uparrow 1$ is equal to the Araki relative entropy \cite{araki1976relative}
\[
D_1(\psi\|\varphi)=\left\{\begin{array}{cc} (h_\psi^{1/2},\log (\Delta_{\psi,\varphi}) h_\psi^{1/2}), & \mbox{ if } s(\psi)\le  s(\varphi)\\ & \\
\infty & \mbox{otherwise}
\end{array}\right.
\]
For $\alpha>1$ it may happen that $h_\varphi^{1/2}\notin \mathcal D(\Delta_{\psi,\varphi}^{\alpha/2})$ and then $D_\alpha(\psi\|\varphi)=\infty$. 
But if $D_{\alpha_0}(\psi\|\varphi)$ is finite for some $\alpha_0>1$, then the function $\alpha\mapsto D_\alpha(\psi\|\varphi)$ is increasing on $(1,\alpha_0]$ and the limit $\lim_{\alpha\downarrow 1} D_\alpha(\psi\|\varphi)=D_1(\psi\|\varphi)$ holds, see 
 \cite[Proposition 11]{berta2018renyi} for a proof of these properties. 

We next find upper and lower bounds for $\tilde D_\alpha$ in terms of $D_\alpha$. For this, we use some upper and lower bounds on the norm $\|h_\psi\|_{p,\varphi}$. The upper bound in the following proposition was proved also in \cite{berta2018renyi} and can be seen as an extension of the Araki-Lieb-Thirring inequality of \cite{kosaki1992aninequality} to non-semifinite case.

\begin{prop}\label{prop:alt} Let $\psi,\varphi\in \Me_*^+$, $s(\psi)\le s(\varphi)$, $p>1$. Then
\[
\psi(1)^{1-p} \|\Delta_{\psi,\varphi}^{1-1/2p}h_\varphi^{1/2}\|^{2p}_2 \le   \|h_\psi\|_{p,\varphi}^p\le \|\Delta_{\psi,\varphi}^{p/2}h_\varphi^{1/2}\|^2.
\]

\end{prop}

\begin{proof} For the first inequality, we may assume that  $h_\psi\in L_p(\Me,\varphi)$, so that $h_\psi= h_\varphi^{1/2q}h_\xi^{1/p}h_\varphi^{1/2q}$ for some $\xi\in \Me_*^+$ and $\|h_\psi\|_{p,\varphi}^p=\xi(1)$. By uniqueness of the polar decomposition in $L_2(\Me)$, we have
\[
h_\xi^{1/2p}h_\varphi^{1/2q}=uh_\psi^{1/2},
\] 
where $u$ is a partial isometry with $u^*u=s(\psi)$. By section \ref{sec:rmo}, $h_\psi^{1/2}\in \mathcal D(\Delta^{1/2q}_{\psi,\varphi})$ and 
\[
\Delta^{1-1/2p}_{\psi,\varphi}h_\varphi^{1/2}=\Delta^{1/2q}_{\psi,\varphi}\Delta_{\psi,\varphi}^{1/2}h_\varphi^{1/2}=\Delta^{1/2q}_{\psi,\varphi}h_\psi^{1/2}= 
h_\psi^{1/2q}u^* h_\xi^{1/2p}.
\]
By H\"older's inequality, 
\[
\|h_\psi^{1/2q}u^* h_\xi^{1/2p}\|_2\le \|h_\psi^{1/2q}\|_{2q}\|h_\xi^{1/2p}\|_{2p} = \psi(1)^{1/2q}\xi(1)^{1/2p}.
\]
This implies the first inequality. For the second inequality, assume that $h_\varphi^{1/2}\in \mathcal D(\Delta_{\psi,\varphi}^{p/2})$ (otherwise the right hand side is infinite and there is nothing to prove). By \cite[Lemma VI.2.3]{takesaki2003TOAII},
it follows that  there is a bounded continuous function $k: S\to L_2(\Me)$, holomorphic in the interior of $S$, given by
\[
 k(z):=\Delta_{\psi,\varphi}^{zp/2}h_\varphi^{1/2}.
 \]
On the other hand, we have $\Delta_{\varphi,\psi}^{-p/2}h_\varphi^{1/2}=J \Delta_{\psi,\varphi}^{p/2}Jh_\varphi^{1/2}=J \Delta_{\psi,\varphi}^{p/2}h_\varphi^{1/2} $, so that
 $h_\varphi^{1/2}\in \mathcal D(\Delta_{\varphi,\psi}^{-p/2})$ and there is a bounded continuous function $k':S\to L_2(\Me)$, holomorphic in the interior of $S$, given by
\[
 k'(z):=\Delta_{\varphi,\psi}^{-zp/2}h_\varphi^{1/2}.
\]
Hence $f(z):= k'(z)k(z)$ defines a bounded continuous function  $S\to L_1(\Me)$, holomorphic in the interior.
By \eqref{eq:app_modular}, we have 
\[
k(z)=h_\psi^{zp/2}h_\varphi^{1/2-zp/2},\quad k'(z)=h_\varphi^{1/2-zp/2}h_\psi^{zp/2},\qquad Re(z)\le 1/p.
\]
It follows that
\[
f(it)=h_\varphi^{1/2}\left(h_\varphi^{-itp/2}h_\psi^{itp}h_\varphi^{-itp/2}\right)h_\varphi^{1/2},\qquad t\in \mathbb R,
\]
so that $f(it)\in L_\infty(\Me,\varphi)$ and $\|f(it)\|_{\infty,\varphi}\le 1$. Further, we obtain by  H\"older's inequality
\[
\|f(1+it)\|_1\le \|k'(1+it)\|_2\|k(1+it)\|_2= \|\Delta_{\psi,\varphi}^{p/2}h_\varphi^{1/2}\|_2^2.
\]
It follows that $f\in \Fe(L_\infty(\Me,\varphi),L_1(\Me))$ and since $f(1/p)=h_\psi$, we obtain by Theorem \ref{thm:htl} that $\|h_\psi\|_{p,\varphi}\le \|\Delta_{\psi,\varphi}^{p/2}h_\varphi^{1/2}\|_2^{2/p}$.

\end{proof}

\begin{coro}\label{coro:standard} Let $\psi,\varphi\in \states_*(\Me)$, $\alpha>1$. Then
\[
D_{2-1/\alpha}(\psi\|\varphi)\le \tilde D_\alpha(\psi\|\varphi)\le D_\alpha(\psi\|\varphi).
\]

\end{coro}

\begin{proof} Immediate from Proposition \ref{prop:alt}.

\end{proof}

\subsubsection{Properties of the function $\alpha\mapsto \tilde D_\alpha$}

\begin{prop}\label{prop:positivity} Let $\psi,\varphi\in \states_*(\Me)$ be such that $h_\psi\in L_\alpha(\Me,\varphi)$ for some $1<\alpha<\infty$.
\begin{enumerate}
\item[(i)]  $\tilde D_\alpha(\psi\|\varphi)\ge 0$, with equality if and only if $\psi=\varphi$.
\item[(ii)] If $\psi\ne \varphi$, the function   $\alpha'\mapsto \tilde D_{\alpha'}(\psi\|\varphi)$ is strictly increasing for $\alpha'\in (1,\alpha]$. 

\end{enumerate}

\end{prop}

\begin{proof} By (\ref{eq:embedding}), $\|h_\psi\|_{\alpha,\varphi}\ge \|h_\psi\|_1=1$, hence  $\tilde D_\alpha(\psi\|\varphi)\ge 0$. Since $h_\varphi=h_\varphi^{1/2}1 h_\varphi^{1/2}\in L_\infty(\Me,\varphi)$ and 
\[
1=\|h_\varphi\|_1\le \|h_\varphi\|_{\alpha,\varphi}\le \|h_\varphi\|_{\infty,\varphi}=\|1\|=1,
\]
we have $\tilde D_\alpha(\varphi\|\varphi)=0$, for all $\alpha$. Assume now that  
$\tilde D_\alpha(\psi\|\varphi)=0$. Choose any  $1<p<\alpha$, then $1\le \|h_\psi\|_{p,\varphi}\le \|h_\psi\|_{\alpha,\varphi}=1$, 
so that  
$\|h_\psi\|_{p,\varphi}=1$. Let $f$ be the constant function  $f(z)\equiv h_\psi$ for all $z\in S$, then clearly 
$f\in \Fe_{\alpha,1}$ (recall the notation of Section \ref{sec:hadamard}). Let $\eta\in (0,1)$ be such that $1/p=\eta+(1-\eta)/\alpha$, then
$f$ satisfies equality in the Hadamard three lines theorem (Theorem \ref{thm:htl})
 at $\eta$. It follows that $h_\psi\equiv f_{h_\psi,p}(z+(1-z)/\alpha)$, $z\in S$ (note that in this case $M=M_1/M_0=1$). 
 Hence $f_{h_\psi,p}(z)\equiv h_\psi$ for all $z\in S$. Putting $z=0$, we get $h_\psi = ch_\varphi$, where $c=\|h_\psi\|_{\alpha,\varphi}=1$.  This finishes the proof of (i).

For (ii), let $1<\alpha'<\alpha''\le \alpha$. Then $h_\psi\in L_{\alpha}(\Me,\varphi)\subseteq L_{\alpha''}(\Me,\varphi)\subseteq L_{\alpha'}(\Me,\varphi)$. Let $\eta$ be such that 
$\eta+(1-\eta)/\alpha{''}=1/\alpha'$. Then $1-\eta=\beta''/\beta'$, where $1/\beta'+1/\alpha'=1$ and $1/\alpha''+1/\beta''=1$. We consider again  the constant function $f(z)\equiv h_\psi$, which this time is an element of  $\Fe_{\alpha'',1}$. By Theorem  
\ref{thm:htl} with $p=1$ and $p'=\alpha''$, we obtain
\[
\|h_\psi\|_{\alpha',\varphi}\le \|h_\psi\|_{\alpha'',\varphi}^{1-\eta}=\|h_\psi\|_{\alpha'',\varphi}^{\beta''/\beta'}
\]
Taking the logarithm  proves $\tilde D_{\alpha'}(\psi\|\varphi)\le \tilde D_{\alpha''}(\psi\|\varphi)$. Assume now that equality holds, then it follows that $h_\psi\equiv f_{h_\psi,\alpha'}(z+(1-z)/\alpha'')M^{z-\eta}$, similarly as in the proof of (i). Putting $u:=z+(1-z)/\alpha''$, we obtain
$h_\psi\equiv f_{h_\psi,\alpha'}(u)M^{au+b}$ for all $u$ with $Re(u)$ between $1/\alpha''$ and $1$, here $a,b\in \mathbb R$. This
 equality clearly extends to all $u\in S$. Again, putting $u=0$, we obtain $h_\psi=ch_\varphi$, which implies $\psi=\varphi$, since both $\varphi$ and $\psi$ are states.

\end{proof}

We next discuss the limit values $\alpha=1,\infty$. Let us define
\[
\tilde D_\infty(\psi\|\varphi):=\left\{ \begin{array}{cc} \log \|h_\psi\|_{\infty,\varphi} & \mbox{ if } h_\psi\in L_\infty(\Me,\varphi)\\ \ &\ \\
\infty & \mbox{otherwise}.
\end{array}\right.
\]
This quantity is clearly an extension of the relative max entropy \eqref{eq:maxent}.

\begin{prop}\label{prop:limits} Let $\psi,\varphi\in \states_*(\Me)$. Then 
\begin{enumerate}
\item[(i)]
$\lim_{\alpha\to \infty}\tilde D_\alpha(\psi\|\varphi)=\tilde D_\infty(\psi\|\varphi)$.
\item[(ii)] If $\tilde D_\alpha(\psi\|\varphi)$ is finite for some $\alpha>1$, then 
\[
\lim_{\alpha\downarrow 1}\tilde D_\alpha(\psi\|\varphi)=D_1(\psi\|\varphi).
\]
\end{enumerate}

\end{prop}

\begin{proof}
First, let $y\in \Me$ and consider the function  $[0,1]\ni \theta\mapsto \log(\|\varphi_y\|_{1/\theta,\varphi})$. This function is decreasing by \eqref{eq:embedding} and 
 by applying Theorem \ref{thm:htl} to the constant function 
$f(z)\equiv \varphi_y$, we see that it is also convex. It follows that this function must be continuous on the interval $(0,1]$. Consequently, we must have 
$\lim_{q\to 1}\|\varphi_y\|_{q,\varphi}=\|\varphi_y\|_1=\varphi(y)$.    

To prove (i), it is enough to show that $\lim_{p\to \infty} \|h_\psi\|_{p,\varphi} =\|h_\psi\|_{\infty,\varphi}$, where we put the 
norms infinite if $h_\psi\notin L_p(\Me,\varphi)$. Note that  the function $p\mapsto \|h_\psi\|_{p,\varphi}$ is increasing and bounded above by $\|h_\psi\|_{\infty,\varphi}$. The statement (i) is clearly true if the limit is infinite, so assume that 
$\lim_{p\to \infty} \|h_\psi\|_{p,\varphi}=M<\infty$. We then have 
\[
\|h_\psi\|_{p,\varphi}\le M\le \|h_\psi\|_{\infty,\varphi}
\]
for all $1\le p<\infty$. Let $y\in \Me^+$. Then for any $q> 1$
 \[
\frac{\<h_y,h_\psi\>}{\|h_y\|_{q,\varphi}}\le \|h_\psi\|_{p,\varphi}\le M, \qquad 1/p+1/q=1,
 \]
hence $\psi(y)=\<h_y,h_\psi\>\le M\|h_y\|_{q,\varphi}$. Taking the limit $q\to 1$, we obtain $\psi(y)\le M\|h_y\|_1=M\varphi(y)$. Since this holds for all $y\in \Me^+$, we obtain  $\psi\le M\varphi$ and by Lemma \ref{lemma:linftyplus}, $h_\psi\in   L_\infty(\Me,\varphi)$, with $\|h_\psi\|_{\infty,\varphi}\le M$. 
The statement (ii) follows from Corollary \ref{coro:standard}  and properties of the standard R\'enyi relative entropy $D_\alpha(\psi\|\varphi)$.

\end{proof}

\subsection{Extension to $\Me_*^+$}

 It is clear from Theorem \ref{thm:lpfi} and the remarks at the beginning of Section \ref{sec:sandwiched} that the spaces $L_p(\Me,\varphi)$ can be defined for $\varphi\in \Me_*^+$ (we put $L_p(\Me,\varphi)=\{0\}$ for $\varphi=0$) and that for $\lambda>0$, $\|h\|_{p,\lambda\varphi}=\lambda^{1/q}\|h\|_{p,\varphi}$ 
for any $h\in L_p(\Me,\varphi)=L_p(\Me,\lambda\varphi)$. The definition of $\tilde D_\alpha$ can thus be extended to  positive normal functionals. It is easy to see that  for $\mu,\lambda>0$ and $\psi,\varphi\in \Me_*^+$, we have 
\begin{equation}\label{eq:multiples}
\tilde D_\alpha(\mu\psi\|\lambda\varphi)=\tilde D_\alpha(\psi\|\varphi)+\frac{\alpha}{\alpha-1}\log \mu-\log\lambda.
\end{equation}
With this extension, we have the following order relations.

\begin{prop}\label{prop:order1} Let $\psi,\psi_0,\varphi,\varphi_0,\in \Me_*^+$ and  $\psi_0\le \psi$, $\varphi_0\le \varphi$. Let 
$1<\alpha<\infty$. Then 
  $\tilde D_\alpha(\psi_0\|\varphi)\le \tilde D_\alpha(\psi\|\varphi)$ and 
  $\tilde D_\alpha(\psi\|\varphi_0)\ge \tilde D_\alpha(\psi\|\varphi)$.

\end{prop}

\begin{proof} For the first inequality, we may assume that $\psi\in L_\alpha(\Me, \varphi)$. The  inequality then  follows by Corollary \ref{coro:positive}. Let now $\varphi_0\le \varphi$ and assume that $h\in L_\infty(\Me,\varphi_0)$. By Lemma \ref{lemma:linfty}, it is easy to see that  then $h\in  L_\infty(\Me,\varphi)$ and 
$\|h\|_{\infty,\varphi}\le \|h\|_{\infty,\varphi_0}$. It follows that if $f\in \Fe(L_\infty(\Me,\varphi_0),L_1(\Me))$, then
$f\in \Fe(L_\infty(\Me,\varphi), L_1(\Me))$ and 
\[
\vertiii{f}_{\Fe(L_\infty(\Me, \varphi), L_1(\Me))}\le \vertiii{f}_{\Fe(L_\infty(\varphi_0),L_1(\Me))}.
\]
By the definition of the interpolation norm, we obtain $\|h_\psi\|_{\alpha,\varphi}\le \|h_\psi\|_{\alpha,\varphi_0}$, 
this implies the second inequality.

\end{proof}

\begin{prop}\label{prop:lsc} $\tilde D_\alpha: \Me_*^+\times \Me_*^+\to [0,\infty]$ is jointly lower semicontinuous.

\end{prop}

\begin{proof} It suffices to prove that the set $\{(\psi,\varphi)\in \Me^+_*\times \Me^+_*,\ \|h_\psi\|_{\alpha,\varphi}\le a\}$ is closed in $\Me_*\times \Me_*$ for each $a\ge 0$. So let $\psi_n$ and $\varphi_n$ be sequences of positive normal functionals, converging in $\Me_*$  to $\psi$ and $\varphi$, respectively, and such that $\|h_{\psi_n}\|_{\alpha,\varphi_n}\le a$.  By Theorem \ref{thm:lpfi}, we have 
\[
h_{\psi_n}=h_{\varphi_n}^{1/2\beta}k_n h_{\varphi_n}^{1/2\beta},\qquad k_n\in L_\alpha(\Me)^+,\ \|k_n\|_\alpha=\|h_{\psi_n}\|_{\alpha,\varphi_n}\le a.
\]
Since the space $L_\alpha(\Me)$ is reflexive and $\{k_n\}$ is bounded, we may assume that $k_n$ converges to some $k$  weakly in $L_\alpha(\Me)$, and then  $\|k\|_\alpha\le a$.

 Let  
$h:=h_\varphi^{1/2\beta}kh_\varphi^{1/2\beta}$, so that $h\in L_\alpha(\Me,\varphi)$ with $\|h\|_{\alpha,\varphi}=\|k\|_\alpha\le a$. We will show that $h_{\psi_n}$ converges to $h$ weakly in $L_1(\Me)$ and hence we must have $h_\psi=h$. 
So let $x\in \Me$. Then by H\"older's inequality,
\begin{align*}
|\Tr (h_{\psi_n}-h)x|&=|\Tr (h_{\varphi_n}^{1/2\beta}k_n h_{\varphi_n}^{1/2\beta}-h_\varphi^{1/2\beta}kh_\varphi^{1/2\beta})x|\\
&\le|\Tr (h_{\varphi_n}^{1/2\beta}-h_\varphi^{1/2\beta})k_n h_{\varphi_n}^{1/2\beta}x| \\
&+ |\Tr h_\varphi^{1/2\beta}k_n(h_{\varphi_n}^{1/2\beta}-h_\varphi^{1/2\beta})x| + |\Tr h_\varphi^{1/2\beta}(k_n-k)h_\varphi^{1/2\beta}x|\\
&\le \|h_{\varphi_n}^{1/2\beta}-h_\varphi^{1/2\beta}\|_{2\beta}\|k_n\|_\alpha(\varphi_n(1)^{1/2\beta}+\varphi(1)^{1/2\beta})\|x\|\\
&+|\Tr (k_n-k)h_\varphi^{1/2\beta}xh_\varphi^{1/2\beta}|
\end{align*}
It was proved by Kosaki \cite[Theorem 4.2]{kosaki1984applicationsuc} that the map 
$L_1(\Me)^+\ni h \mapsto h^{1/p}\in L_p(\Me)^+$ is norm continuous. Hence the first part of the last expression  converges to 0. Since $h_\varphi^{1/2\beta}xh_\varphi^{1/2\beta}\in L_\beta(\Me)$ for any $x\in \Me$, the second part  goes to 0 as well.

\end{proof}

Let now $\mathcal N=\Me\oplus \Me$ and $\varphi_1,\varphi_2\in \Me_*^+$ be faithful,  $\varphi=\varphi_1\oplus\varphi_2$. 
By \cite{terp1981lpspaces}, $L_p(\Ne)=L_p(\Me)\times L_p(\Me)$ and $\|(k_1,k_2)\|_p=(\|k_1\|_p^p+\|k_2\|_p^p)^{1/p}$, $1\le p\le \infty$. By this and Theorem \ref{thm:lpfi}, we obtain that $L_p(\Ne,\varphi)=L_p(\Me,\varphi_1)\times L_p(\Me,\varphi_2)$ and 
for $h=(h_1,h_2)\in L_p(\Ne, \varphi)$,
\begin{equation}\label{eq:oplus}
\|(h_1,h_2)\|_{p,\varphi}=(\|h_1\|_{p,\varphi_1}^p+\|h_2\|_{p,\varphi_2}^p)^{1/p}.
\end{equation}

\begin{prop}\label{prop:genmean} Let $\psi_1,\psi_2,\varphi_1,\varphi_2\in \Me_*^+$ and let $\psi=\psi_1\oplus \psi_2$, 
$\varphi=\varphi_1\oplus\varphi_2$. Then
\begin{align*}
\exp\{(\alpha-1)\tilde D_\alpha(\psi\|\varphi)\}=\exp\{(\alpha-1)\tilde D_\alpha&(\psi_1\|\varphi_1)\}\\
 &+\exp\{(\alpha-1)\tilde D_\alpha(\psi_1\|\varphi_1)\}.
\end{align*}

\end{prop}

\begin{proof} Follows immediately from \eqref{eq:oplus} and the definition of $\tilde D_\alpha$.

\end{proof}

\subsection{Data processing inequality}

Let $\Ne$ be another von Neumann algebra and let $\Phi: L_1(\Me)\to L_1(\Ne)$ be a  positive linear trace-preserving map. Then $\Phi$ defines a positive linear map $\Me_*\to \Ne_*$, also denoted by $\Phi$, mapping states to states. 
The  adjoint $\Phi^*:\Ne\to\Me$ is normal, positive and  unital. The map $\Phi$ will be fixed throughout this section, together with $\varphi\in \states_*(\Me)$. We put  $e:=s(\varphi)$ and $e':=s(\Phi(\varphi))$.  

We first show that $\Phi$ maps $L_1(\Me,\varphi)$ into $L_1(\Ne,\Phi(\varphi))$, see the remarks at the beginning of Section \ref{sec:sandwiched}. From  $\varphi(\Phi^*(1-e'))=\Phi(\varphi)(1-e')=0$, it follows that 
$e\Phi^*(1-e')e=0$ and hence $e\Phi^*(e')=e$, so that $ e\le\Phi^*(e')$. Let now $h=ehe\in L_1(\Me)^+$, then
\[
\Tr h=\Tr he\le \Tr h\Phi^*(e')=\Tr \Phi(h)e'\le \Tr \Phi(h)=\Tr h,
\]
hence $e'\Phi(h)e'=\Phi(h)$ and $\Phi(h)\in L_1(\Ne, \Phi(\varphi))$. Since $L_1(\Me,\varphi)$ is generated by positive elements, this implies that $\Phi$ maps $L_1(\Me,\varphi)$ into $L_1(\Ne,\Phi(\varphi))$. 

Assume next that $h=h_x$ for some $x\in e\Me^+e$. Then $h_x\le \|x\|h_\varphi$ and since 
$\Phi$ is positive, we also have $\Phi(h_x)\le\|x\|\Phi(h_\varphi)$. By Lemma \ref{lemma:linftyplus}, there is some $x'\in e'\Ne^+e'$ such that
\[ 
\Phi(h_x)=\Phi(h_\varphi^{1/2}xh_\varphi^{1/2})=\Phi(h_\varphi)^{1/2}x'\Phi(h_\varphi)^{1/2}=\Phi(h_\varphi)_{x'}\in L_\infty(\Ne,\Phi(\varphi))^+.
\] 
Since $\Me^+$ generates $\Me$, it follows that $\Phi$ maps $L_\infty(\Me,\varphi)$ into $L_\infty(\Ne,\Phi(\varphi))$. By linearity, the map 
 $x\mapsto x'$ extends to a  linear map $\Phi^*_\varphi: e\Me e\to e'\Ne e'$, which is obviously positive, unital and normal.

\begin{prop}\label{prop:contraction}
 For any  $1\le p\le \infty$, $\Phi$ restricts to a contraction $L_p(\Me, \varphi)\to L_p(\Ne, \Phi(\varphi))$. 
\end{prop}

\begin{proof} As we have seen, $\Phi$ maps $L_1(\Me,\varphi)$ into $L_1(\Ne,\Phi(\varphi))$ and $L_\infty(\Me,\varphi)$ into $L_\infty(\Ne,\Phi(\varphi))$.
For any $h\in L_1(\Me,\varphi)$, 
\[
\|\Phi(h)\|_1=\sup_{x_0\in \Ne , \|x_0\|\le1} \Tr \Phi(h)x_0=\sup_{x_0\in \Ne , \|x_0\|\le1} \Tr h\Phi^*(x_0)\le \|h\|_1,
\]
the last inequality follows from the fact that $\Phi^*$ is a unital positive map, hence a contraction by the Russo-Dye theorem, \cite{paulsen2002completely}.  Next, for $x\in e\Me e$,
\[
\|\Phi(h_x)\|_{\infty,\Phi(\varphi)}=\|\Phi(h_\varphi)_{\Phi^*_\varphi(x)}\|_{\infty,\Phi(\varphi)}= \|\Phi^*_\rho(x)\|\le \|x\|=\|h_x\|_{\infty,\rho},
\] 
where we used Russo-Dye theorem for $\Phi^*_\varphi$.  The statement now follows by the Riesz-Thorin theorem (Theorem \ref{thm:rt}).

\end{proof}

Let us denote the preadjoint of $\Phi_\varphi^*$ by $\Phi_\varphi$. For any $x\in e \Me e$ and  $h_0\in e'L_1(\Ne) e'$, we have
\begin{equation}\label{eq:pread1}
\<h_0,\Phi(h_x)\>=\<h_0, \Phi(h_\varphi)_{\Phi_\varphi^*(x)}\>=
\Tr h_0\Phi_\varphi^*(x)=\<\Phi_\varphi(h_0),h_x\>.
\end{equation}
By the uniqueness part in \cite[Theorem 4.4.1]{belo1976interpolation}, this extends to 
\begin{equation}\label{eq:dualp}
\<h_0,\Phi(h)\>=\<\Phi_\varphi(h_0),h\>,\qquad h\in L_p(\Me, \varphi),\ h_0\in L_q(\Ne, \Phi(\varphi)).
\end{equation}
Moreover, for $x\in e\Me e$,
\begin{align*}
\Tr \Phi_\varphi(\Phi(h_\varphi))x&= \Tr \Phi(h_\varphi)\Phi_\varphi^*(x)=
\Tr \Phi(h_\varphi)^{1/2}\Phi_\varphi^*(x)\Phi(h_\varphi)^{1/2}\\
&= \Tr \Phi(h_x)=\Tr h_x=\Tr h_\varphi x
,
\end{align*}
so that $\Phi_\varphi(\Phi(h_\varphi))=h_\varphi$. By Proposition \ref{prop:contraction}, $\Phi_\varphi$  defines a positive contraction $L_p(\Ne,\Phi(\varphi))\to L_p(\Me, \varphi)$, for $1\le p\le \infty$.

\begin{rem}\label{rem:petzdual}  As in the proof of Lemma \ref{lemma:linftyplus},
 we have
\[
\Tr h_xy=( yh^{1/2}_\varphi, Jxh^{1/2}_\varphi),\qquad y\in e\Me e,\ x\in  e\Me ^+e
\]
 and by linearity, this holds for all $x\in e\Me e$. 
It follows that  $\Phi^*_\varphi$ is determined by
\begin{align*}
(\Phi^*(y_0)h_\varphi^{1/2},Jxh^{1/2}_\varphi )&=\Tr h_x\Phi^*(y_0)=\Tr \Phi(h_x)y_0\\
&=\Tr \Phi(h_\varphi)^{1/2}\Phi^*_\varphi(x)\Phi(h_\varphi)^{1/2} y_0\\
&=(y_0 \Phi(h_\varphi)^{1/2},J_0\Phi^*_\varphi(x)\Phi(h_\varphi)^{1/2})
 \end{align*}
for all $y_0\in e'\Ne e'$ and $x\in e\Me e$, here $J_0$ is the modular conjugation (adjoint operation) on $L_2(e' \Ne e')$. In this form, the map $\Phi^*_\varphi$ was defined by Petz in \cite{petz1988sufficiency} and is therefore called the Petz dual. Moreover, it was proved  that for any $n$, $\Phi^*_\varphi$ is $n$-positive if and only if $\Phi$ is. 
\end{rem}

We are  now ready to prove the data processing inequality for $\tilde D_\alpha$, together with some lower and upper bounds in terms of the 
dual elements $T_{\beta,\varphi}(h_\psi)$ and $T_{\beta,\Phi(\varphi)}(\Phi(h_\psi))$, see \eqref{eq:dual}.

\begin{thm}\label{thm:dtp} Let $1<\alpha<\infty$, $1/\alpha+1/\beta=1$. Let 
$\psi,\varphi\in \states_*(\Me)$ and assume that  $h_\psi\in L_\alpha(\Me,\varphi)$. Let us denote $h:=T_{\beta,\varphi}(h_\psi)$, $h_0=T_{\beta,\Phi(\varphi)}(\Phi(h_\psi))$. Then
 for  $1<\alpha\le 2$,
\begin{align*}
\tilde D_\alpha(\psi\|\varphi)-\tilde D_\alpha(\Phi(\psi)\|\Phi(\varphi))\ge 2\|\frac12(h-\Phi_\varphi(h_0))\|_{\beta,\varphi}^\beta 
\end{align*}
and for  $2\le \alpha<\infty$,
\begin{align*}
 \tilde D_\alpha(\psi\|\varphi)-\tilde D_\alpha(\Phi(\psi)\|\Phi(\varphi))\ge  \beta(\beta-1)\|\frac12(h-\Phi_\varphi(h_0))\|_{\beta,\varphi}^2.
\end{align*}
If $1<\alpha<\infty$ and $\|h-\Phi_\varphi(h_0)\|_{\beta,\varphi}<1$, we also have an upper bound
\[
\tilde D_\alpha(\psi\|\varphi)-\tilde D_\alpha(\Phi(\psi)\|\Phi(\varphi))\le -
\beta\log\left({1-\|h-\Phi_\varphi(h_0)\|_{\beta,\varphi}}\right).
\]

\end{thm}

\begin{proof}  By \eqref{eq:dualp},  we obtain
\begin{align*}
\frac{\|\Phi(h_\psi)\|_{\alpha,\Phi(\varphi)}}{\|h_\psi\|_{\alpha,\varphi}}&=\frac{\<h_0,\Phi(h_\psi)\>}{\|h_\psi\|_{\alpha,\varphi}}
=\<\Phi_\varphi(h_0),\|h_\psi\|_{\alpha,\varphi}^{-1}h_\psi\>\\&=\<\Phi_\varphi(h_0)+h,\|h_\psi\|_{\alpha,\varphi}^{-1}h_\psi\>-
\<h,\|h_\psi\|_{\alpha,\varphi}^{-1}h_\psi\> \\
&\le \|\Phi_\varphi(h_0)+h\|_{\beta,\varphi}-1.
\end{align*}
Assume $1<\alpha\le 2$, so that $2\le \beta<\infty$. Since $\|h\|_{\beta,\varphi},\|\Phi_\varphi(h_0)\|_{\beta,\varphi}\le 1$,   Clarkson's inequality (Theorem \ref{thm:clarkson}) implies 
\[
\|\Phi_\varphi(h_0)+h\|_{\beta,\varphi}\le (2^{\beta}-\|h-\Phi_\varphi(h_0)\|_{\beta,\varphi}^\beta)^{1/\beta}= 2(1-\|\frac12(h-\Phi_\varphi(h_0)\|_{\beta,\varphi}^\beta)^{1/\beta}.
\]
Using the inequality $(1-x^p)^{1/p}\le 1-\frac1{p}x^p$ for $p>1$, $x\in [0,1]$, we obtain
\[
\|\Phi_\varphi(h_0)+h\|_{\beta,\varphi}-1\le 1-\frac2{\beta}\|\frac12(h-\Phi_\varphi(h_0)\|_{\beta,\varphi}^\beta
\]
For $2\le \alpha<\infty$, we apply Theorem \ref{thm:pixu} with $h$ replaced by $h+\Phi_\varphi(h_0)$ and  $k$ by $h-\Phi_\varphi(h_0)$, and obtain 
\[
\|\Phi_\varphi(h_0)+h\|_{\beta,\varphi}\le 2(1-(\beta-1)\|\frac12(h-\Phi_\varphi(h_0)\|_{\beta,\varphi}^2)^{1/2}.
\]
The inequality above with $p=2$ now yields
\[
\|\Phi_\varphi(h_0)+h\|_{\beta,\varphi}-1\le 1-(\beta-1)\|\frac12(h-\Phi_\varphi(h_0)\|_{\beta,\varphi}^2
\]
The inequalities in (i) and (ii)  follow by taking the logarithms and using the inequality $\log x\le x-1$ for $x>0$.

On the other hand, we have a lower bound 
\begin{align*}
\frac{\|\Phi(h_\psi)\|_{\alpha,\Phi(\varphi)}}{\|h_\psi\|_{\alpha,\varphi}}&=\<\Phi_\varphi(h_0),\|h_\psi\|_{\alpha,\varphi}^{-1}h_\psi\>=\<h-(h-\Phi_\varphi(h_0)),\|\psi\|_{\alpha,\varphi}^{-1}h_\psi\>\\
&\ge 1-\|h-\Phi_\varphi(h_0)\|_{\beta,\varphi}.
\end{align*}
If $1-\|h-\Phi_\varphi(h_0)\|_{\beta,\varphi}> 0$, this implies (iii).

\end{proof}

The following result was obtained in \cite{mhre2015monotonicity} for algebras of bounded operators on a separable Hilbert space.

\begin{coro} Let $\psi,\varphi\in \states_*(\Me)$ and let $\Phi:L_1(\Me)\to L_1(\Ne)$ be a positive trace preserving map. Then
\[
D_1(\Phi(\psi)\|\Phi(\varphi))\le D_1(\psi\|\varphi).
\]

\end{coro}

\begin{proof} Immediate from Theorem \ref{thm:dtp} and 
 Proposition \ref{prop:limits}.

\end{proof}

\begin{coro}\label{coro:convex}
For $1<\alpha<\infty$, the map $(\psi,\varphi)\mapsto \exp\{(\alpha-1)\tilde D_\alpha\}$ is jointly convex.

\end{coro}

\begin{proof} The following arguments are standard.
Let $\psi_1,\psi_2,\varphi_1,\varphi_2\in \states_*(\Me)$. Let $\psi,\varphi\in \states_*(\Me\oplus\Me)$ be given by $\psi=\lambda\psi_1\oplus(1-\lambda)\psi_2$ and $\varphi=\lambda\varphi_1\oplus(1-\lambda)\varphi_2$. By Proposition \ref{prop:genmean} and \eqref{eq:multiples}, we obtain
\begin{align*}
\exp\{(\alpha-1)\tilde D_\alpha(\psi\|\varphi)\}= &\exp\{(\alpha-1)\tilde D_\alpha(\lambda\psi_1\|\lambda\varphi_1)\}\\
&+ \exp\{(\alpha-1)\tilde D_\alpha((1-\lambda)\psi_2\|
(1-\lambda)\varphi_2)\} \\
=& \lambda\exp\{(\alpha-1)\tilde D_\alpha(\psi_1\|\varphi_1)\}\\
&+
(1-\lambda)\exp\{(\alpha-1)\tilde D_\alpha(\psi_2\|\varphi_2)\}.
\end{align*}
Let $\Phi:L_1(\Me\oplus \Me)\to L_1(\Me)$ be given by $(h_1,h_2)\mapsto h_1+h_2$, then $\Phi$ is obviously positive and trace preserving and 
\[
\Phi(\varphi)=\lambda\varphi_1+(1-\lambda)\varphi_2,\quad \Phi(\psi)=\lambda\psi_1+(1-\lambda)\psi_2.
\]
The statement now follows by Theorem \ref{thm:dtp}. 
\end{proof}

We also obtain a characterization of equality, which will be useful in the next section.

\begin{coro}\label{coro:equality} 
Let $\psi,\varphi\in \states_*(\Me)$ and assume that $\psi\in L_\alpha(\Me,\varphi)$.  Then 
$\tilde D_\alpha(\psi\|\varphi)=\tilde D_\alpha(\Phi(\psi)\|\Phi(\varphi))$
if and only if 
\[
\Phi_\varphi\circ T_{\beta,\Phi(\varphi)}\circ \Phi(h_\psi)=T_{\beta,\varphi}(h_\psi).
\] 
If $\alpha=2$, this is equivalent to $\Phi_\varphi\circ \Phi(\psi)=\psi$. 
\end{coro}

\begin{proof} The first statement is immediate from Theorem \ref{thm:dtp}. 
Let now $\alpha=2$, then 
\begin{align*}
\|\Phi(h_\psi)\|_{2,\Phi(\varphi)}^2=\<\Phi(h_\psi),\Phi(h_\psi)\>&=\<h_\psi, \Phi_\varphi\circ\Phi(h_\psi)\>\\
&\le
\|h_\psi\|_{2,\varphi}\|\Phi_\varphi\circ\Phi(h_\psi)\|_{2,\varphi}\le \|h_\psi\|^2_{2,\varphi}.
\end{align*}
The statement now follows by equality condition in the Schwarz inequality.
 
 \end{proof}

\section{Sufficiency of channels}

In this section, we study the case of equality in DPI for $\tilde D_\alpha$. The aim is to show that this equality implies  existence of a recovery map for $(\Phi,\psi,\varphi)$. For this, we need that the map $\Phi$ is 2-positive, which will be assumed in the rest of the paper. 

Let $\psi,\varphi\in \states_*(\Me)$ and let $\Phi: L_1(\Me)\to L_1(\Ne)$ be a 2-positive trace preserving map. We say that $\Phi$ is sufficient with respect to $\{\psi,\varphi\}$ if there exists a 2-positive trace preserving recovery map $\Psi: L_1(\Ne)\to L_1(\Me)$, such that $\Psi\circ\Phi(h_\psi)=h_\psi$ and $\Psi\circ\Phi(h_\varphi)=h_\varphi$. 

\begin{rem} In the above definition, we may also assume that both $\Phi$ and $\Psi$ are completely positive and trace preserving maps, such maps are usually called quantum channels. This definition seems stronger, but in fact it is fully equivalent, in the sense that if $\Phi$ is 2-positive and trace preserving and there is a 2-positive recovery map $\Psi$ for $(\Phi,\psi,\varphi)$, then there are quantum channels $\tilde \Phi$ and $\tilde \Psi$ that coincide with $\Phi$ and $\Psi$ when restricted to $\{\psi,\varphi\}$ and $\{\Phi(\psi),\Phi(\varphi)\}$, respectively.

\end{rem}

The following theorem is one of the crucial results of \cite{petz1988sufficiency}. Note that it implies that $\Phi_\varphi$ is a universal recovery map. 

\begin{thm}\label{thm:petzrecovery}\cite{petz1988sufficiency, jepe2006sufficiency}
Let $\Phi: L_1(\Me)\to L_1(\Ne)$ be a trace preserving 2-positive map. Let  $\varphi\in \states_*(\Me)$ be faithful and assume that $\Phi(\varphi)$ is faithful as well.  Then  $\Phi$ is sufficient with respect to $\{\psi,\varphi\}$ if and only if
$\Phi_\varphi\circ\Phi(h_\psi)=h_\psi$.

\end{thm}

The following is a standard result of ergodic theory.

\begin{lemma}\label{lemma:condexp} Let $\Omega:L_1(\Me)\to L_1(\Me)$ be 2-positive and trace preserving, admitting a faithful normal invariant state.  Then there is a faithful normal conditional expectation $E$ on $\Me$ such that  $\psi\in \states_*(\Me)$ is invariant under $\Omega$ if and only if 
$\psi\circ E=\psi$.
\end{lemma}

\begin{proof}
Let $\Se$ be the set of all normal invariant states of $\Omega$ and let $\mathcal I$ be the set of all 2-positive unital normal  maps 
$T:\Me\to\Me$, such that $\psi\circ T=\psi$ for all $\psi\in \Se$. Then $\mathcal I$ is a semigroup (i.e. closed under composition), convex and closed with respect to the pointwise weak*-topology.  By the mean ergodic theorem \cite{kuna1979mean},  $\mathcal I$ contains a conditional expectation $E$, such that 
\[
T\circ E=E\circ T=E,\qquad \forall T\in \mathcal I.
\]
Since $E\in \mathcal I$, $\psi\circ E=\psi$ for all $\psi\in \Se$. On the other hand, let $\psi\in \states_*(\Me)$ be such that $\psi\circ E=\psi$, then
\[
\psi\circ \Omega^*=\psi\circ E\circ \Omega^*=\psi\circ E=\psi,
\]
because $\Omega^*\in \mathcal I$.

\end{proof}

\begin{lemma}\label{lemma:sufficpsiomega} 
Let $\varphi\in\states_*(\Me)$ be faithful.  Let $1<p<\infty$ and let $\psi\in \states_*(\Me)$ be such that
\[
h_\psi=c h_\varphi^{1/2q}h_\omega^{1/p}h_\varphi^{1/2q}
\]
for some $c>0$ and $\omega\in \states_*(\Me)$. Let $\Phi:L_1(\Me)\to L_1(\Ne)$ be a 2-positive trace preserving map such that $\Phi(\varphi)$ is faithful. Then $\Phi$ is sufficient with respect to $\{\psi,\varphi\}$ if and only if it is sufficient with respect to $\{\omega,\varphi\}$.

\end{lemma}

\begin{proof} Let $\Omega=\Phi_\varphi\circ \Phi$, then $\varphi$ is a faithful invariant state for $\Omega$. By Lemma \ref{lemma:condexp} and Theorem \ref{thm:petzrecovery}, there is a faithful normal conditional expectation $E$ such that $\varphi\circ E=\varphi$ and
$\Phi$ is sufficient with respect to $\{\psi,\varphi\}$ if and only if $\psi\circ E=\psi$. Let us denote the range of $E$ by $\Me_0$.

We now apply the results in Appendix \ref{sec:ce}. Let  
$\psi\circ E=\psi$, that is, $E_1(h_\psi)=h_\psi$.  By \eqref{eq:condexp_holder} and \eqref{eq:ce_L1},
\[
h_\psi=E_1(h_\psi)=c E_1(h_\varphi^{1/2q}h_\omega^{1/p}h_\varphi^{1/2q})=
c h_\varphi^{1/2q}E_p(h_\omega^{1/p})h_\varphi^{1/2q}.
\]
Since  $i_p$ is an isomorphism (see Theorem \ref{thm:lpfi}), we see that we must have 
$h^{1/p}_\omega=E_p(h^{1/p}_\omega)\in L_p(\Me_0)$. But then also $h_\omega\in L_1(\Me_0)$, so that 
$\omega\circ E=\omega$ and $\Phi$ is sufficient with respect to $\{\omega,\varphi\}$. Conversely, if $\omega\circ E=\omega$, then
 $h_\omega^{1/p}\in L_p(\Me_0)$, so that  
$h_\psi\in L_1(\Me_0)$ and  $\psi\circ E=\psi$.

\end{proof}

\begin{lemma}\label{lemma:norm} Let $\Phi:L_1(\Me)\to L_1(\Ne)$ be a positive trace preserving map and let $1<p<\infty$. Let  $h\in L_p(\Me,\varphi)$ be such that 
$\|\Phi(h)\|_{p,\Phi(\varphi)}=\|h\|_{p,\varphi}$.  Then  
\[
\|\Phi(f_{p,h}(\theta))\|_{1/\theta,\Phi(\varphi)}=\|f_{p,h}(\theta)\|_{1/\theta,\varphi},\quad \forall \theta\in (0,1).
\]

\end{lemma}

\begin{proof} By Proposition \ref{prop:contraction}, $\Phi\circ f_{p,h}\in \Fe(L_\infty(\Ne,\Phi(\varphi)),L_1(\Ne))=:\Fe_0$ and $\vertiii{\Phi\circ f_{p,h}}_{\Fe_0}\le \vertiii{f_{p,h}}_\Fe$. Since  
$\Phi\circ f_{p,h}(1/p)=\Phi(h)$, we have 
\[
\|\Phi(h)\|_{p,\Phi(\varphi)}\le \vertiii{\Phi\circ f_{p,h}}_{\Fe_0}\le \vertiii{f_{p,h}}_\Fe=\|h\|_{p,\varphi}=\|\Phi(h)\|_{p,\Phi(\varphi)},
\]
hence $\|\Phi\circ f_{p,h}(1/p))\|_{p,\Phi(\varphi)}= \vertiii{\Phi(f_{p,h})}_{\Fe_0}=\vertiii{f_{p,h}}_\Fe$. 
The result now follows by Lemma \ref{lemma:equal}.

\end{proof}

We are now prepared to prove the main result of this section.

\begin{thm}\label{thm:sufficiency} Let $\Phi:L_1(\Me)\to L_1(\Ne)$ be a 2-positive trace preserving map and let $1<\alpha<\infty$. Let $\varphi,\psi\in \states_*(\Me)$ be such that $h_\psi\in L_\alpha(\Me,\varphi)$. Then $\Phi$ is sufficient with respect to $\{\psi,\varphi\}$ if and only if $\tilde D_\alpha(\psi\|\varphi)=\tilde D_\alpha(\Phi(\psi)\|\Phi(\varphi))$.

\end{thm}

\begin{proof} By the assumptions, $s(\psi)\le s(\varphi)$ and we may suppose that both $\varphi$ and $\Phi(\varphi)$ are faithful, by restriction to the corresponding compressed algebras. Further, we have
$h_\psi=h_\varphi^{1/2\beta}h_\omega^{1/\alpha}h_\varphi^{1/2\beta}$ for some $\omega\in \Me_*^+$, here $1/\alpha+1/\beta=1$.  Suppose that $\tilde D_\alpha(\psi\|\varphi)=\tilde D_\alpha(\Phi(\psi)\|\Phi(\varphi))$. Then 
$\|\Phi(h_\psi)\|_{\alpha,\Phi(\varphi)}=\|h_\psi\|_{\alpha,\varphi}$ and by Lemma  \ref{lemma:norm}, 
\[
\|\Phi(f_{\alpha,h_\psi}(1/2))\|_{2,\Phi(\varphi)}=\|f_{\alpha,h_\psi}(1/2)\|_{2,\varphi}.
\]
Note that 
\[
f_{\alpha,h_\psi}(1/2)=ch_\varphi^{1/4}h_\omega^{1/2}h_\varphi^{1/4}\in L_1(\Me)^+
\]
for some constant $c>0$, hence  there is  some $\psi_1\in \states_*(\Me)$, such that 
$f_{\alpha,h_\psi}(1/2)=dh_{\psi_1}$, where $d>0$ is obtained by normalization. It follows that 
$h_{\psi_1}\in L_2(\Me,\varphi)$ and we have
\[
\|\Phi(h_{\psi_1})\|_{2,\Phi(\varphi)}=\|h_{\psi_1}\|_{2,\varphi}.
\] 
By Corollary \ref{coro:equality}, this implies that $\Phi$ is sufficient with respect to $\{\psi_1,\varphi\}$ and by Lemma \ref{lemma:sufficpsiomega}, $\Phi$ is sufficient with respect to  $\{\omega_1,\varphi\}$, where $\omega_1=\omega(1)^{-1}\omega$. Using  Lemma \ref{lemma:sufficpsiomega} again, we obtain that $\Phi$ is sufficient with respect to $\{\psi,\varphi\}$.

The converse statement follows immediately from DPI (Theorem \ref{thm:dtp}).

\end{proof}

\section{Concluding remarks}

In this paper, an  extension of the sandwiched R\'enyi relative $\alpha$-entropies
to the setting of von Neumann algebras is defined for $\alpha>1$,  using 
 an interpolating family of  non-commutative $L_p$-spaces with respect to a state.  
For this extension, we proved that it coincides with the previously defined Araki-Masuda divergences \cite{berta2018renyi}
Further, some of the basic properties are shown, 
 in particular the data processing inequality with respect to positive trace-preserving maps. Since the limit $\alpha\to 1$ yields the Araki relative entropy $D_1$, this implies that $D_1$ is monotone under such maps  and not only adjoints of unital Schwarz maps, as 
 previously known \cite{uhlmann1977relative}. For $\Me =B(\Ha)$, this fact was recently observed in \cite{mhre2015monotonicity}.

Another main result of the paper is the fact that preservation of the extended sandwiched entropies characterizes 
 sufficiency of 2-positive trace preserving maps. Note that for most of the proofs 2-positivity was not needed, indeed,  Lemma \ref{lemma:condexp} is the only place where more than positivity is necessary. It would be interesting to see whether 
 similar results can be proved assuming only positivity, since the results  known so far  
on sufficiency of maps need stronger positivity conditions. Note that for $\alpha=2$, an extension to positive maps is proved in  Corollary \ref{coro:equality}.

The Araki-Masuda divergences were defined in \cite{berta2018renyi} also for  $\alpha\in [1/2,1)$. A treatment of $\tilde D_\alpha$ for these values in our setting will be given elsewhere, see \cite{jencova2017renyi2}. 

\section*{Acknowledgement}

I am grateful to Fumio Hiai for useful discussions and for sharing his notes about equality of the two extensions of sandwiched R\'enyi entropies. 

\appendix

\renewcommand{\thesection}{\Alph{section}}
\setcounter{equation}{0}
\renewcommand{\theequation}{\thesection.\arabic{equation}}

\newtheorem{thmap}{Theorem}[section]

\section{Some technical results in Haagerup $L_p$-spaces}

\subsection{Relative modular operator}\label{sec:rmo}
We discuss the definition of the relative modular operator and its form in the standard representation $(\lambda(\Me), L_2(\Me), J, L_2(\Me)^+)$. 

Let $\eta,\xi\in L_2(\Me)$ and let $\psi=\omega_\eta=(\cdot\,\eta,\eta)$, $\varphi=\omega_\xi=(\cdot\,\xi,\xi)$. The conjugate-linear operator $S_{\eta,\xi}$ with domain $\Me \xi+L_2(\Me)(1-s(\varphi))$ is  defined as 
\begin{equation}\label{eq:app_S}
S_{\eta,\xi}:\, x\xi+\zeta \mapsto s(\varphi)x^*\eta,\qquad x\in \Me,\ \zeta \in L_2(\Me)(1-s(\varphi)).
\end{equation}
Let also $F_{\eta,\xi}$ be defined on the domain $\xi\Me +(1-s(\varphi))L_2(\Me)$ as
\begin{equation}\label{eq:app_F}
F_{\eta,\xi}:\, \xi y+\zeta' \mapsto \eta y^*s(\varphi),\qquad y\in \Me,\ \zeta' \in (1-s(\varphi))L_2(\Me).
\end{equation}
Then  $S_{\eta,\xi}$, $F_{\eta,\xi}$ are densely defined  and closable, and we have $\bar S=F^*$, $\bar F=S^*$. The closures have polar decompositions
\[  
\bar S_{\eta,\xi}=J_{\eta,\xi}\Delta^{1/2}_{\eta,\xi},\quad \bar F_{\eta,\xi}=\Delta^{1/2}_{\eta,\xi}J_{\xi,\eta}=J_{\xi,\eta}\Delta_{\xi,\eta}^{-1/2},
\]
 where $J_{\eta,\xi}$ is a partial anti-isometry, $J_{\xi,\eta}=J_{\eta,\xi}^*$ and $\Delta_{\eta,\xi}$ is a positive self-adjoint operator on $L_2(\Me)$, called the relative modular operator. This operator  does not depend on 
the choice of the vector representative $\eta$ of $\psi$ and we may replace $J_{\eta,\xi}$ with $J$ if $\eta,\xi\in L_2(\Me)^+$ (which means that $\eta=h_\psi^{1/2}$, $\xi=h_\varphi^{1/2}$).  See e.g. \cite[Appendix C]{arma1982positive} and \cite{takesaki2003TOAII} for more details. 
We use the notation $\Delta_{\psi,\varphi}:=\Delta_{\eta,h_\varphi^{1/2}}$.

Note that for $z=\alpha+it$, $0\le \alpha\le 1/2$, $t\in \mathbb R$, we have \cite{kosaki1981positive}
\begin{align*}
\mathcal D(\Delta_{\psi,\varphi}^z)=\mathcal D(\Delta_{\psi,\varphi}^\alpha)&=\{k\in L_2(\Me),\ h_\psi^{\alpha}k h_\varphi^{-\alpha}\in L_2(\Me)\}\\
&=\{k\in L_2(\Me), \exists k'\in L_2(\Me), h_\psi^{\alpha}ks(\varphi)=k'h_\varphi^\alpha\}
\end{align*}
and for $k\in \mathcal D(\Delta_{\psi,\varphi}^z)$, 
\begin{equation}\label{eq:app_modular}
\Delta_{\psi,\varphi}^zk=h_\psi^z k h_\varphi^{-z}=h_\psi^{it}k' h_\varphi^{-it}.
\end{equation} 
Moreover, since  $J\Delta_{\psi,\varphi}J=\Delta_{\varphi,\psi}^{-1}$, we have  $\mathcal D(\Delta_{\varphi,\psi}^{-z})=\mathcal D(\Delta_{\varphi,\psi}^{-\alpha})=J\mathcal D(\Delta_{\psi,\varphi}^{\alpha})$
and for $k\in \mathcal D(\Delta_{\psi,\varphi}^{z})$, 
\begin{equation}\label{eq:app_modular_d}
\Delta_{\varphi,\psi}^{-z}k^*= h_\varphi^{-z}k^* h_\psi^z
=h_\varphi^{-it}(k')^*h_\psi^{it} .
\end{equation}

\subsection{The spatial derivative}\label{sec:spatial}

We now recall the definition of the  spatial derivative $\Delta(\eta/\varphi)$ of 
\cite{berta2018renyi} in the above standard representation.  
Let $\Ha_\varphi:=[\Me h_\varphi^{1/2}]=L_2(\Me)s(\varphi)$ and let $\xi\in L_2(\Me)$ be such that the corresponding functional  is  majorized  by $\varphi$: 
\[
\omega_\xi(a^*a)=\|a\xi\|^2\le C_\xi\varphi(a^*a),\qquad \forall a\in \Me,
\]
for some positive constant $C_\xi$. Then
\[
R^\varphi(\xi): ah_\varphi^{1/2}\mapsto a\xi,\qquad a\in \Me
\] 
extends to a bounded linear operator $\Ha_\varphi\to L_2(\Me)$. Obviously, $R^\varphi(\xi)$ extends to a bounded linear operator on $L_2(\Me)$ by putting it equal to 0 on $L_2(\Me)(1-s(\varphi))$. Moreover, this operator commutes with  the left action of $\Me$, so that it belongs to  $\lambda(\Me)'=\rho(\Me)$, where $\rho$ is the right action $\rho(a): h\mapsto ha$, $h\in L_2(\Me)$.  In fact, $\omega_\xi$ is majorized by $\varphi$ if and only if $\xi\in h_\varphi^{1/2}\Me$, so  that there is some $y_\xi\in \Me$ such that $\xi=h_\varphi^{1/2}y_\xi$, $s(\varphi)y_\xi=
y_\xi$  and we have
$R^\varphi(\xi)=\rho(y_\xi)$.    

Let now $\eta\in L_2(\Me)$. The spatial derivative $\Delta(\eta/\varphi)$ is a positive self-adjoint operator associated with the quadratic form 
$\xi\mapsto (\eta, R^\varphi(\xi)R^\varphi(\xi)^*\eta)$ as
\begin{align*}
(\xi,\Delta(\eta/\varphi)\xi)&=(\Delta(\eta/\varphi)^{1/2}\xi,\Delta(\eta/\varphi)^{1/2}\xi)=(\eta, R^\varphi(\xi)R^\varphi(\xi)^*\eta)\\
&=(R^\varphi(\xi)^*\eta,R^\varphi(\xi)^*\eta)=(\eta y_\xi^*s(\varphi),\eta y_\xi^*s(\varphi))=
(F_{\eta,h_\varphi^{1/2}}\xi,F_{\eta,h_\varphi^{1/2}}\xi),
\end{align*}
see \eqref{eq:app_F}. Since $h_\varphi^{1/2}\Me+ (1-s(\varphi))L_2(\Me)$ is a core for both 
$\Delta(\eta/\varphi)$ and $F_{\eta,h_\varphi^{1/2}}$, it follows that 
\[
\Delta(\eta/\varphi)= F_{\eta,h_\varphi^{1/2}}^*F_{\eta,h_\varphi^{1/2}}=J\Delta_{\omega,\varphi}J,
\]
where $\omega:=\omega_\eta$.
This implies that for any $\xi\in L_2(\Me)$ and $\gamma\in \mathbb C$, we have   
\begin{equation}\label{eq:app_spatial}
\|\Delta(\eta/\varphi)^\gamma \xi\|_2=
\|\Delta_{\omega,\varphi}^\gamma J\xi\|_2=\|\Delta_{\omega,\varphi}^\gamma \xi^*\|_2.
\end{equation}

\subsection{Extensions of conditional expectations} \label{sec:ce}

 A conditional expectation $E$ on a von Neumann algebra $\Me$ is a positive contractive normal projection onto a von Neumann subalgebra $\Me_0\subseteq \Me$. A conditional expectation is necessarily completely positive and satisfies the condition
 \begin{equation}\label{eq:app_ce}
E(xay)=xE(a)y,\qquad x,y\in \Me_0,\ a\in \Me.
 \end{equation}
Assume a faithful normal state $\phi$ and a von Neumann subalgebra $\Me_0\subseteq \Me$ are given, 
such that there is a  conditional expectation $E$ satisfying $\phi\circ E=\phi$. Then 
the space $L_p(\Me_0)$ for $1\le p\le\infty$ can be identified with a subspace  in $L_p(\Me)$ and 
$E$ can be extended  to a  contractive projection  $E_p$ of $L_p(\Me)$ onto $L_p(\Me_0)$, \cite{juxu2003noncommutative}.
This  extension is  positive and satisfies
\begin{equation}\label{eq:condexp_holder}
E_s(hlk)=hE_r(l)k,\qquad h\in L_p(\Me_0), k\in L_q(\Me_0), l\in L_r(\Me),
\end{equation}
whenever  $1\le p,q,r\le \infty$ are such that $1/p+1/q+1/r=1/s\le 1$. 
Moreover, for  $p=1$, we have 
\begin{equation}\label{eq:ce_L1}
E_1: h_\psi\mapsto h_{\psi\circ E},\qquad \psi\in \Me_*.
\end{equation}

\section{The complex interpolation method}\label{sec:interp}

 In this paragraph, we  briefly describe the complex interpolation method, following \cite{belo1976interpolation}, see also \cite{kosaki1984applications}.

Let $(X_0,X_1)$ be a compatible pair of Banach spaces, with norms $\|\cdot\|_0$ and $\|\cdot\|_1$. For our purposes, it is enough to assume that $X_0$ is continuously embedded in $X_1$. Let $S\subset \mathbb C$ be the strip $S=\{z\in \mathbb C,\ 0\le Re(z)\le 1\}$ and let $\Fe=\Fe(X_0,X_1)$ be the set of functions $f:S\to X_1$ such that
\begin{enumerate}
\item[(a)] $f$ is bounded, continuous on $S$ and analytic in the interior of $S$
\item[(b)] For $t\in \mathbb R$, $f(it)\in X_0$ and the map 
  $t\in \mathbb R \mapsto f(it) \in X_0$ is continuous and bounded.
\end{enumerate}
For $f\in \Fe$, let 
\[
\vertiii{f}_\Fe=\max\{\sup_t \|f(it)\|_{0}, \sup_t\|f(1+it)\|_1\}
\]
Then $(\Fe,\vertiii{\cdot}_\Fe)$ is a Banach space. For $0<\theta<1$, the  interpolation space is defined as the set
\[
 C_{\theta}(X_0,X_1)=\{f(\theta), f\in \Fe\}
\]
endowed with the norm
\begin{equation}\label{eq:theta}
\|x\|_{\theta}=\inf\{\vertiii{f}_\Fe,\ f(\theta)=x,\ f\in \Fe\}.
\end{equation}

Since $C_\theta(X_0,X_1)$ is the quotient space $\Fe/K_\theta$ with respect to the closed subspace $K_\theta=\{f\in \Fe, f(\theta)=0\}$, we see that $C_\theta(X_0,X_1)$ is a Banach space. Moreover, we have the continuous embeddings
\[
X_0\subseteq C_\theta(X_0,X_1)\subseteq X_1
\]
 and $C_\theta$ defines an exact interpolation functor of exponent $\theta$, which means that  the following abstract version of the Riesz-Thorin interpolation theorem holds.

\begin{thmap}\label{thm:rt} Let $(X_0,X_1)$ and $(Y_0,Y_1)$ be pairs of compatible Banach spaces and let $T:X_1\to Y_1$ be a bounded linear operator such that $T(X_0)\subseteq Y_0$. If $\|Tx\|_{Y_1}\le M_1\|x\|_{X_1}$, $x\in X_1$ and  $\|Tx_0\|_{X_0}\le M_0\|x_0\|_{X_0}$ for $x_0\in X_0$, then for $\theta\in (0,1)$, 
\[
\|Tx\|_\theta\le M_0^{1-\theta}M_1^\theta\|x\|_\theta.
\]

\end{thmap}


\end{document}